\documentclass{article}

\usepackage{amsthm}

\usepackage{amssymb,amsmath,amsfonts}
\usepackage{enumerate}
\usepackage{algorithm,algpseudocode}
\usepackage{mathtools}
\usepackage{graphicx}
\usepackage[usenames]{color}
\usepackage{xspace}
\usepackage{tikz,pgfplots}
\usepackage{fullpage}

\newcommand{\mjrty}{\textsf{MJRTY}\xspace}
\newcommand{\misrag}{\textsf{MisraGries}\xspace}
\newcommand{\countsketch}{\textsf{CountSketch}\xspace}
\newcommand{\countmin}{\textsf{CountMin}\xspace}
\newcommand{\countsieve}{\textsf{CountSieve}\xspace}
\newcommand{\lossy}{\textsf{LossyCounting}\xspace}
\newcommand{\spacesave}{\textsf{SpaceSaving}\xspace}
\newcommand{\frequent}{\textsf{Frequent}\xspace}
\newcommand{\pickdrop}{\textsf{Pick-and-Drop}\xspace}
\newcommand{\ceil}[1]{\lceil #1 \rceil}
\newcommand{\eps}{\varepsilon}
\renewcommand{\epsilon}{\varepsilon}
\newcommand{\bptree}{\textsf{BPTree}\xspace}
\newcommand{\hho}{\textsf{HH1}\xspace}
\newcommand{\hht}{\textsf{HH2}\xspace}

\newcommand{\inprod}[1]{\left\langle #1 \right\rangle}

\newcommand{\bfone}{\mathbf{1}}

\DeclareMathOperator{\poly}{poly}

\DeclareMathOperator{\median}{median}

\newcommand{\supbound}[2]{c #1\beta^{#2}}

\newcommand{\calC}{\mathcal{C}}
\newcommand{\E}{\mathbb{E}}

\newcommand{\calH}{\mathcal{H}}
\newcommand{\R}{\mathbb{R}}
\newcommand{\calT}{\mathcal{T}}
\newcommand{\Z}{\mathbb{Z}}

\newtheorem{theorem}{Theorem}
\newtheorem{lemma}[theorem]{Lemma}

\newtheorem{proposition}[theorem]{Proposition}

\begin{document}

\title{BPTree: an $\ell_2$ heavy hitters algorithm using constant memory}
\author{
 Vladimir Braverman\thanks{\texttt{vova@cs.jhu.edu}. Johns Hopkins University. Supported in part by NSF grants IIS-1447639 and CCF-1650041 and by a Google Faculty Research Award.}
\and
Stephen R.~Chestnut\thanks{\texttt{stephenc@ethz.ch}. ETH Zurich}
\and
Nikita Ivkin\thanks{\texttt{nivkin1@jhu.edu}. Johns Hopkins University. Supported in part by NSF grants IIS-1447639 and CCF-1650041 and by DARPA grant N660001-1-2-4014.}
\and 
Jelani Nelson\thanks{\texttt{minilek@seas.harvard.edu}. Harvard University. Supported by NSF grant IIS-1447471 and CAREER CCF-1350670, ONR Young Investigator award N00014-15-1-2388, and a Google Faculty Research Award.}
\and
Zhengyu Wang\thanks{\texttt{zhengyuwang@g.harvard.edu}. Harvard University. Supported in part by NSF grant CCF-1350670.}
\and
David P.~Woodruff\thanks{\texttt{dwoodruf@cs.cmu.edu}. Carnegie Mellon University. This work was done while a research staff member at IBM Research.}
}

\maketitle

\begin{abstract}
The task of finding heavy hitters is one of the best known and well studied problems in the area of data streams.
One is given a list $i_1,i_2,\ldots,i_m\in[n]$ and the goal is to identify the items among $[n]$ that appear frequently in the list.
In sub-polynomial space, the strongest guarantee available is the $\ell_2$ guarantee, which requires finding all items that occur at least $\epsilon\|f\|_2$ times in the stream, where the vector~$f\in\R^n$ is the count histogram of the stream with $i$th coordinate equal to the number of times~$i$ appears~$f_i\coloneqq\#\{j\in[m]:i_j=i\}$.
The first algorithm to achieve the $\ell_2$ guarantee was the \countsketch of \cite{CharikarCF04}, which requires $O(\epsilon^{-2}\log n)$ words of memory and $O(\log n)$ update time and is known to be space-optimal if the stream allows for deletions. The recent work of \cite{BravermanCIW16} gave an improved algorithm for insertion-only streams, using only $O(\epsilon^{-2}\log\epsilon^{-1}\log\log n)$ words of memory.
In this work, we give an algorithm \bptree for $\ell_2$ heavy hitters in insertion-only streams that achieves $O(\epsilon^{-2}\log\epsilon^{-1})$ words of memory and $O(\log\epsilon^{-1})$ update time, which is the optimal dependence on $n$ and $m$.
In addition, we describe an algorithm for tracking $\|f\|_2$ at all times with $O(\epsilon^{-2})$ memory and update time.
Our analyses rely on bounding the expected supremum of a Bernoulli process involving Rademachers with limited independence, which we accomplish via a Dudley-like chaining argument that may have applications elsewhere.
\end{abstract}

\section{Introduction}

The {\em streaming model} of computation is well-established as one important model for processing massive datasets. 
A sequence of data is seen which could be, for example, destination IP addresses of TCP/IP packets or query terms submitted to a search engine, and which is modeled as a list of integers $i_1,i_2,\ldots,i_m\in[n]$.
Each item $i\in[n]$ has a \emph{frequency} in the stream that is the number of times it appears and is denoted $f_i\coloneqq\#\{j\in[m]:i_j=i\}$.
The challenge is to define a system to answer some pre-defined types of queries about the data, such as distinct element counts, quantiles, frequent items, or other statistics of the vector $f$, to name a few. 
The system is allowed to read the stream only once, in the order it is given, and it is typically assumed that the stream being processed is so large that explicitly storing it is undesirable or even impossible. 
Ideally, streaming algorithms should use space sublinear, or even exponentially smaller than, the size of the data to allow the algorithm's memory footprint to fit in cache for fast stream processing. 
The reader is encouraged to read \cite{BabcockBDMW02,Muthukrishnan05} for further background on the streaming model of computation.

Within the study of streaming algorithms, the problem of finding frequent items is one of the most well-studied and core problems, with work on the problem beginning in 1981 \cite{BoyerM81,BoyerM91}. 
Aside from being an interesting problem in its own right, algorithms for finding frequent items are used as subroutines to solve many other streaming problems, such as moment estimation \cite{IndykW05}, entropy estimation \cite{ChakrabartiCM10,HarveyNO08}, $\ell_p$-sampling \cite{MonemizadehW10}, finding duplicates \cite{GopalanR09}, and several others.

Stated simply, the goal is to report a list of items that appear least $\tau$ times, for a given threshold $\tau$.
Naturally, the threshold $\tau$ should be chosen to depend on some measure of the size of the stream.
The point of a frequent items algorithm is to highlight a \emph{small} set of items that are frequency outliers.
A choice of $\tau$ that is independent of $f$ misses the point; it might be that all frequencies are larger than $\tau$.

With this in mind, previous work has parameterized $\tau$ in terms of different norms of $f$ with \misrag~\cite{MisraG82} and \countsketch~\cite{CharikarCF04} being two of the most influential examples. 
A value $\epsilon>0$ is chosen, typically $\epsilon$ is a small constant independent of $n$ or $m$, and $\tau$ is set to be $\epsilon\|f\|_1=\epsilon m$ or $\epsilon\|f\|_2$.
These are called the $\ell_1$ and $\ell_2$ guarantees, respectively.
Choosing the threshold $\tau$ in this manner immediately limits the focus to outliers since no more than $1/\epsilon$ items can have frequency larger than $\epsilon\|f\|_1$ and no more than $1/\epsilon^2$ can have frequency $\epsilon\|f\|_2$ or larger.

A moments thought will lead one to conclude that the $\ell_2$ guarantee is stronger, i.e.\ harder to achieve, than the $\ell_1$ guarantee because $\|x\|_1\geq\|x\|_2$, for all $x\in\R^n$.
Indeed, the $\ell_2$ guarantee is much stronger.
Consider a stream with all frequencies equal to 1 except one which is equal $j$.
With $\epsilon=1/3$, achieving the $\ell_1$ guarantee only requires finding an item with frequency~$j= n/2$, which means that it occupies more than one-third of the positions in the stream, whereas achieving the $\ell_2$ guarantee would require finding an item with frequency $j=\sqrt{n}$, such an item is a negligible fraction of the stream!

As we discuss in Section~\ref{sec: prev work}, the algorithms achieving $\ell_2$ guarantee, like \countsketch~\cite{CharikarCF04}, achieve essentially the best space-to-$\tau$ trade-off.
But, since the discovery of \countsketch, which uses $O(\epsilon^{-2}\log{n})$ words of memory\footnote{Unless explicitly stated otherwise, space is always measured in machine words. 
It is assumed a machine word has at least $\log_2 \max\{n,m\}$ bits, to store any ID in the stream and the length of the stream.}, it has been an open problem to determine the smallest space possible for achieving the $\ell_2$ guarantee.
Since the output is a list of up to $\epsilon^{-2}$ integers in $[n]$, $\Omega(\epsilon^{-2})$ words of memory are necessary.
The recent algorithm \countsieve shrunk the storage to $O(\epsilon^{-2}\log\epsilon^{-1}\log\log n)$ words~\cite{BravermanCIW16}.
Our main contribution is to describe an algorithm~\bptree that uses only $O(\epsilon^{-2}\log\epsilon^{-1})$ words---independent of $n$ and $m$!---and obviously optimal in the important setting where $\epsilon$ is a constant.

\subsection{Previous work}\label{sec: prev work}

Work on the heavy hitters problem began in 1981 with the \mjrty algorithm of \cite{BoyerM81,BoyerM91}, which is an algorithm using only two machine words of memory that could identify an item whose frequency was strictly more than half the stream. 
This result was generalized by the \misrag algorithm in \cite{MisraG82}, which, for any $0 < \eps \le 1/2$, uses $2(\ceil{1/\eps} - 1)$ counters to identify every item that occurs strictly more than an $\eps m$ times in the stream. 
This data structure was rediscovered at least two times afterward \cite{DemaineLM02,KarpSP03} and became also known as the \frequent algorithm. 
It has implementations that use $O(1/\eps)$ words of memory, $O(1)$ expected update time to process a stream item (using hashing), and $O(1/\eps)$ query time to report all the frequent items. 
Similar space requirements and running times for finding $\eps$-frequent items were later achieved by the \spacesave \cite{MetwallyAA05} and \lossy \cite{MankuM12} algorithms. 

A later analysis of these algorithms in \cite{BerindeCIS09} showed that they not only identify the heavy hitters, but they also provided estimates of the frequencies of the heavy hitters.
Specifically, when using $O(k/\eps)$ counters they provide, for each heavy hitter $i\in[n]$, an estimate~$\tilde{f}_i$ of the frequency $f_i$ such that $|\tilde{f}_i - f_i| \le (\eps/k)\cdot \|f_{tail(k)}\|_1 \leq (\eps/k)\|f\|_1$.  
Here $f_{tail(k)}$ is the vector $f$ but in which the largest $k$ entries have been replaced by zeros (and thus the norm of $f_{tail(k)}$ can never be larger than that of $f$). 
We call this the \emph{$((\eps/k), k)$-tail guarantee}.
A recent work of \cite{BhattacharyyaDW16} shows that for $0<\alpha < \eps \le 1/2$, all $\eps$-heavy hitters can be found together with approximate for them $\tilde{f}_i$ such that $|\tilde{f}_i - f_i| \le \alpha \|f\|_1$, and the space complexity is $O(\alpha^{-1}\log(1/\eps) + \eps^{-1}\log n + \log\log\|f\|_1)$ bits.

All of the algorithms in the previous paragraph work in one pass over the data in the {\em insertion-only} model, also known as the {\em cash-register} model \cite{Muthukrishnan05}, where deletions from the stream are not allowed. Subsequently, many algorithms have been discovered that work in more general models such as the strict turnstile and general turnstile models. 
In the turnstile model, the vector $f\in\R^n$ receives updates of the form $(i,\Delta)$, which triggers the change $f_i \leftarrow f_i + \Delta$; note that we recover the insertion-only model by setting $\Delta=1$ for every update. 
The value $\Delta$ is assumed to be some bounded precision integer fitting in a machine word, which can be either positive or negative. In the {\em strict} turnstile model we are given the promise that $f_i \ge 0$ at all times in the stream. That is, items cannot be deleted if they were never inserted in the first place. In the general turnstile model no such restriction is promised (i.e.\ entries in $f$ are allowed to be negative). This can be useful when tracking differences or changes across streams. For example, if $f^1$ is the query stream vector with $f^1_i$ being the number of times word $i$ was queried to a search engine yesterday, and $f^2$ is the similar vector corresponding to today, then finding heavy coordinates in the vector $f = f^1 - f^2$, which corresponds to a sequence of updates with $\Delta = +1$ (from yesterday) followed by updates with $\Delta = -1$ (from today), can be used to track changes in the queries over the past day. 

In the general turnstile model, an $\eps$-heavy hitter in the $\ell_p$ norm is defined as an index $i\in[n]$ such that $|f_i| \ge \eps\|f\|_p$. Recall $\|f\|_p$ is defined as $(\sum_{i=1}^n |f_i|^p)^{1/p}$. 
The \countmin sketch treats the case of $p=1$ and uses $O(\eps^{-1}\log n)$ memory to find all $\eps$-heavy hitters and achieve the $(\eps, 1/\eps)$-tail guarantee
\cite{CormodeM05}. 
The \countsketch treats the case of $p=2$ and uses $O(\eps^{-2}\log n)$ memory, achieving the $(\eps, 1/\eps^2)$-tail guarantee. It was later showed in \cite{JowhariST11} that the \countsketch actually solves $\ell_p$-heavy hitters for all $0<p\le 2$ using $O(\eps^{-p}\log n)$ memory and achieving the $(\eps, 1/\eps^p)$-tail guarantee. 
In fact they showed something stronger: that {\em any} $\ell_2$ heavy hitters algorithm with error parameter $\eps^{p/2}$ achieving the tail guarantee automatically solves the $\ell_p$ heavy hitters problem with error parameter $\eps$ for any $p\in (0,2]$. 
In this sense, solving the heavy hitters for $p=2$ with tail error, as \countsketch does, provides the strongest guarantee among all $p\in (0,2]$. 

Identifying $\ell_2$ heavy hitters is optimal in another sense, too. 
When $p>2$ by H\"older's Inequality $\epsilon\|f\|_p\geq \frac{\epsilon}{n^{1/2-1/p}}\|f\|_2$.
Hence, one can use an $\ell_2$ heavy hitters algorithm to identify items with frequency at least $\epsilon\|f\|_p$, for $p>2$, by setting the heaviness parameter of the $\ell_2$ algorithm to $\epsilon/n^{1/2-1/p}$.
The space needed to find $\ell_p$ heavy hitters with a \countsketch is therefore $O(\epsilon^{-2}n^{1-2/p}\log n)$ which is known to be optimal~\cite{li2013tight}.
We conclude that the $\ell_2$ guarantee leads to the best space-to-frequency-threshold ratio among all $p>0$.

It is worth pointing out that both the \countmin sketch and \countsketch are randomized algorithms, and with small probability $1/n^c$ (for a user specified constant $c>0$), they can fail to achieve their stated guarantees. The work \cite{JowhariST11} also showed that the \countsketch algorithm is optimal: they showed that {\em any} algorithm, even in the strict turnstile model, solving $\ell_p$ heavy hitters even with $1/3$ failure probability must use $\Omega(\eps^{-p}\log n)$ memory.

The reader may also recall the \pickdrop algorithm of \cite{braverman2013approximating} for finding $\ell_p$ heavy hitters, $p\geq 3$, in insertion-only streams.
\pickdrop uses $O(n^{1-2/p})$ words, so it's natural to wonder whether the same approach would work for $\ell_2$ heavy hitters in $O(1)$ memory.
However, \pickdrop breaks down in multiple, fundamental ways that seem to prevent any attempt to repair it, as we describe in 
Appendix~\ref{app: pick and drop}. 
In particular, for certain streams it has only polynomially small probability to correctly identify an $\ell_2$ heavy hitter.

Of note is that the \misrag and other algorithms in the insertion-only model solve $\ell_1$ heavy hitters using (optimal) $O(1/\eps)$ memory, whereas the \countmin and \countsketch algorithms use a larger $\Theta(\eps^{-1}\log n)$ memory in the strict turnstile model, which is optimal in that model. Thus there is a gap of $\log n$ between the space complexities of $\ell_1$ heavy hitters in the insertion-only and strict turnstile models.
\cite{BravermanCIW16} recently showed that a gap also exists for $\ell_2$ heavy hitters.

The paper~\cite{BravermanCIW16} also introduced a $(1\pm\epsilon)$-relative error $F_2$ tracking scheme based on a linear sketch that uses only $O(\log m\log\log m)$ bits of memory, for constant $\epsilon$.
This is known to be optimal when $n = (\log m)^{O(1)}$ by a lower bound of~\cite{huang2014tracking}.

\subsection{Our contributions}

We provide a new one-pass algorithm, \bptree, which in the insertion-only model solves $\ell_2$ heavy hitters and achieves the $(\eps, 1/\eps^2)$-tail guarantee. 
For any constant $\eps$ our algorithm only uses a constant $O(1)$ words of memory, which is optimal. 
This is the first optimal-space algorithm for $\ell_2$ heavy hitters in the insertion-only model for constant $\eps$.
The algorithm is described in Theorem~\ref{thm: like CountSketch}.

En route to describing \bptree and proving its correctness we describe another result that may be of independent interest.
Theorem~\ref{thm: sup bound} is a new limited randomness supremum bound for Bernoulli processes.  
Lemma~\ref{lem: F2 tracker chaining} gives a more advanced analysis of the algorithm of Alon, Matias, and Szegedy (AMS) for approximating $\|f\|_2$~\cite{AlonMS99}, showing that one can achieve the same (additive) error as the AMS algorithm \emph{at all points in the stream}, at the cost of using 8-wise independent random signs rather than 4-wise independent signs.
An alternative is described in Appendix~\ref{app: another F2} where we show that if one accepts an additional $\log1/\epsilon$ space then 4-wise independent signs suffice.
Note that~\cite{BravermanCIW16} describes an algorithm using $O(\log\log n)$ words that does $F_2$ tracking in an insertion only stream with a multiplicative error $(1\pm\epsilon)$.
The multiplicative guarantee is stronger, albeit with more space for the algorithm, but the result can be recovered as a corollary to our additive $F_2$ tracking theorem, which has a much simplified algorithm and analysis compared to~\cite{BravermanCIW16}.

After some preliminaries, Section~\ref{sec: algo} presents both algorithms and their analyses. 
The description of \bptree is split into three parts.
Section~\ref{sec: chaining} states and proves the chaining inequality.
Section~\ref{sec: experiments} presents the results of some numerical experiments.

\subsection{Overview of approach}\label{sec: overview of approach}

Here we describe the intuition for our heavy hitters algorithm in the case of a single heavy hitter $H\in[n]$ such that $f_H^2\geq\frac{9}{10}\|f\|_2^2$. The reduction from multiple heavy hitters to this case is standard. 
Suppose also for this discussion we knew a constant factor approximation to $F_2 \coloneqq \|f\|_2^2$. 
Our algorithm and its analysis use several of the techniques developed in~\cite{BravermanCIW16}.
We briefly review that algorithm for comparison.

Both \countsieve and \bptree share the same basic building block, which is a subroutine that tries to identify one bit of information about the identity of $H$.
The one-bit subroutine hashes the elements of the stream into two buckets, computes one \emph{Bernoulli process} in each bucket, and then compares the two values.
The Bernoulli process is just the inner product of the frequency vector with a vector of Rademacher (i.e., uniform $\pm1$) random variables.
The hope is that the Bernoulli process in the bucket with $H$ grows faster than the other one, so the larger of the two processes reveals which bucket contains $H$.
In order to prove that the process with $H$ grows faster, \cite{BravermanCIW16} introduce a chaining inequality for insertion-only streams that bounds the supremum of the Bernoulli processes over all times.
The one-bit subroutine essentially gives us a test that $H$ will pass with probability, say, at least $9/10$ and that any other item passes with probability at most $6/10$.
The high-level strategy of both algorithms is to repeat this test sequentially over the stream.

\countsieve uses the one-bit subroutine in a two part strategy to identify $\ell_2$ heavy hitters with $O(\log\log n)$ memory.  
The two parts are (1) amplify the heavy hitter so $f_H\geq (1-\frac{1}{\poly(\log n)})\|f\|_2$ and (2) identify $H$ with independent repetitions of the one-bit subroutine.
Part (1) winnows the stream from, potentially, $n$ distinct elements to at most $n/\poly(\log n)$ elements.
The heavy hitter remains and, furthermore, we get $f_H\geq (1-\frac{1}{\poly(\log n)})\|f\|_2$ because many of the other elements are removed.
\countsieve accomplishes this by running $\Theta(\log\log n)$ independent copies of the one-bit subroutine in parallel, and discarding elements that do not pass a super-majority of the tests.
A standard Chernoff bound implies that only $n/2^{O(\log\log n)}=n/\poly(\log n)$ items survive.
Part (2) of the strategy identifies $\Theta(\log n)$ `break-points' where $\|f\|_2$ of the winnowed stream increases by approximately a $(1+1/\log n)$ factor from one break-point to the next.
Because $H$ already accounts for nearly all of the value of $\|f\|_2$ it is still a heavy hitter within each of the $\Theta(\log n)$ intervals.
\countsieve learns one bit of the identity of $H$ on each interval by running the one-bit subroutine.
After all $\Theta(\log n)$ intervals are completed the identity of $H$ is known.
 
\bptree merges the two parts of the above strategy.  
As above, the algorithm runs a series of $\Theta(\log n)$ rounds where the goal of each round is to learn one bit of the identity of $H$. 
The difference from \countsieve is that \bptree discards more items after every round, then recurses on learning the remaining bits.
As the algorithm proceeds, it discards more and more items and $H$ becomes heavier and heavier in the stream. 
This is reminiscent of work on adaptive compressed sensing \cite{IndykPW11}, but here we are able to do everything in a single pass given the insertion-only property of the stream. 
Given that the heavy hitter is even heavier, it allows us to weaken our requirement on the two counters at the next level in the recursion tree: we now allow their suprema to deviate even further from their expectation, and this is precisely what saves us from having to worry that one of the $O(\log n)$ Bernoulli processes that we encounter while walking down the tree will have a supremum which is too large and cause us to follow the wrong path. 
The fact that the heavy hitter is even heavier also allows us to ``use up'' even fewer updates to the heavy hitter in the next level of the tree, so that overall we have enough updates to the heavy hitter to walk to the bottom of the tree. 

\section{Preliminaries}

An insertion only stream is a list of items $p_1,\ldots,p_m\in[n]$.
The frequency of $j$ at time $t$ is $f_j^{(t)}\coloneqq\#\{i\leq t\mid p_i=j\}$, $f^{(t)}\in\Z_{\geq0}^n$ is called the \emph{frequency vector}, we denote $f\coloneqq f^{(m)}$, $F_2^{(t)} = \sum_{i=1}^n (f_i^{(t)})^2$, $F_2=\sum_{j=1}^nf_j^2$, and $F_0 =\#\{j\in[n]: f_j>0\}$.
An item $H\in[n]$ is a $\alpha$-heavy hitter\footnote{This definition is in a slightly different form from the one given in the introduction, but this form is more convenient when $f_H^2$ is very close to $F_2$.} if $f_H^2\geq \alpha^2 \sum_{j\neq H}f_j^2 = \alpha^2(F_2-f_H^2)$.
For $W\subseteq[n]$, denote by $f^{(t)}(W)\in\Z_{\geq0}^n$ the frequency vector at time $t$ of the stream restricted to the items in $W$, that is, a copy of $f^{(t)}$ with the $i$th coordinate replaced by 0 for every $i\notin W$.
We also define $f^{(s:t)}(W)\coloneqq f^{(t)}(W) - f^{(s)}(W)$ and $F_2(W) = \sum_{j\in W}f_j^2$.
In a case where the stream is semi-infinite (it has no defined end) $m$ should be taken to refer to the time of a query of interest.
When no time is specified, quantities like $F_2$ and $f$ refer to the same query time $m$.

Our algorithms make use of 2-universal (pairwise independent), 4-wise independent, and $8$-wise independent hash functions.
We will commonly denote such a function $h:[n]\to[p]$ where $p$ is a prime larger than $n$, or we may use $h:[n]\to\{0,1\}^R$, which may be taken to mean a function of the first type for some prime $p\in[2^{R-1},2^{R})$.
We use $h(x)_i$ to denote the $i$th bit, with the first bit most significant (big-endian).
A crucial step in our algorithm involves comparing the bits of two values $a,b\in[p]$.
Notice that, for any $0\leq r \leq \lceil\log_2 p\rceil$, we have $a_i=b_i$, for all $1\leq i\leq r$, if and only if $|a-b|<2^{\lceil\log_2p\rceil-r}$.
Therefore, the test $a_i=b_i$, for all $1\leq i\leq r$, can be performed with a constant number of operations.

We will use, as a subroutine, and also compare our algorithm against \countsketch~\cite{CharikarCF04}.
To understand our results, one needs to know that \countsketch has two parameters, which determine the number of ``buckets'' and ``repetitions'' or ``rows'' in the table it stores.
The authors of~\cite{CharikarCF04} denote these parameters $b$ and $r$, respectively. 
The algorithm selects, independently, $r$ functions $h_1,\ldots,h_r$ from a 2-universal family with domain $[n]$ and range $[b]$ and $r$ functions $\sigma_1,\ldots,\sigma_r$ from a 2-universal family with domain $[n]$ and range $\{-1,1\}$.
\countsketch stores the value $\sum_{j:h_t(j)=i} \sigma(j)f_j$, in cell $(t,i)\in[r]\times[b]$ of the table.

In our algorithm we use the notation $\bfone(A)$ denote the indicator function of the event $A$.
Namely, $\bfone(A)=1$ if $A$ is true and 0 otherwise.
We sometimes use $x\lesssim y$ to denote $x=O(y)$.

\section{Algorithm and analysis}\label{sec: algo}

We will now describe and analyze the main algorithm, which is broken into several subroutines.
The most important subroutine is \hho, Algorithm~\ref{algo: HH}, which finds a single $O(1)$-heavy hitter assuming we have an estimate $\sigma$ of $\sqrt{F_2}$ such that $\sqrt{F_2}\leq\sigma\leq2\sqrt{F_2}$.
Next is \hht, Algorithm~\ref{algo: guess}, which  removes the assumption entailing $\sigma$ by repeatedly ``guessing'' values for $\sigma$ and restarting \hho as more items arrive.
The guessing in \hht is where we need $F_2$ tracking.
Finally, a well known reduction from finding $\epsilon$-heavy hitters to finding a single $O(1)$-heavy hitter leads us to the main heavy hitters algorithm \bptree, which is formally described in Theorem~\ref{thm: like CountSketch}.

This section is organized as follows. 
The first subsection gives an overview of the algorithm and its analysis. 
Section~\ref{sec: chaining} proves the bound on the expected supremum of the Bernoulli processes used by the algorithm.
Section~\ref{sec: HH1} uses the supremum bound to prove the correctness of the main subroutine~\hho.
Section~\ref{sec: ams sup} establishes the correctness of $F_2$ tracking.
The subroutine~\hht, which makes use of the $F_2$ tracker, and the complete algorithm~\bptree are described and analyzed in Section~\ref{sec: main algo proof}.

\subsection{Description of the algorithm}
\begin{algorithm}[t]
  \begin{algorithmic}
    \Procedure{\hho}{$\sigma$, $p_1$, $p_2$,\,\ldots,$p_m$}
    \State $R\gets 3\lfloor\log_2(\min\{n,\sigma^2\}+1)\rfloor$
    \State Initialize $b = b_1b_2\cdots b_R=0\in[2^R]$  
    \State Sample $h:[n]\to\{0,1\}^R\sim 2\text{-wise indep.\ family}$ 
    \State Sample $Z\in\{-1,1\}^n$ $4$-wise indep. 
    \State $X_0,X_1\gets 0$ 
    \State $r\gets 1$, $H\gets -1$
    \State $\beta\gets 3/4$, $c\gets 1/32$
    \For {$t=1,2,\ldots,m$ and $r<R$}
      \If {$h(p_t)_i = b_i$, for all $i\leq r-1$} 
        \State $H\gets p_t$
        \State $X_{h(p_t)_r}\gets X_{h(p_t)_r} + Z_{p_t}$ 
        \If {$|X_{0}+X_1|\geq\supbound{\sigma}{r}$}
          \State Record one bit $b_{r}\gets \bfone(|X_1|>|X_0|)$
          \State Refresh $(Z_i)_{i=1}^n$, $X_0,X_1\gets 0$
          \State $r\gets r+1$
        \EndIf
      \EndIf
    \EndFor
    \State \Return $H$ 
    \EndProcedure
  \end{algorithmic}
  \caption{Identify a heavy hitter.}\label{algo: HH}
\end{algorithm}

The crux of the problem is to identify one $K$-heavy hitter for some constant~$K$.
\hho, which we will soon describe in detail, accomplishes that task given a suitable approximation $\sigma$ to $\sqrt{F_2}$.
\hht, which removes the assumption of knowing an approximation $\sigma \in [\sqrt{F}_2, 2\sqrt{F}_2]$, is described in Algorithm~\ref{algo: guess}. 
The reduction from finding all $\epsilon$-heavy hitters to finding a single $K$-heavy hitter is standard from the techniques of \countsketch; it is described in Theorem~\ref{thm: like CountSketch}.

\hho, Algorithm~\ref{algo: HH}, begins with randomizing the item labels by replacing them with pairwise independent values on $R=\Theta(\log \min\{n,\sigma^2\})$ bits, via the hash function $h$.
Since $n$ and $\sigma^2\geq F_2$ are both upper bounds for the number of distinct items in the stream, $R$ can  be chosen so that every item receives a distinct hash value.

Once the labels are randomized, \hho proceeds in rounds wherein one bit of the randomized label of the heavy hitter is determined during each round.
It completes all of the rounds and outputs the heavy hitter's identity within one pass over the stream.
As the rounds proceed, items are discarded from the stream.
The remaining items are called \emph{active}.
When the algorithm discards an item it will never reconsider it (unless the algorithm is restarted).
In each round, it creates two Bernoulli processes $X_0$ and $X_1$.
In the $r$th round, $X_0$ will be determined by the active items whose randomized labels have their $r$th bit equal to 0, and $X_1$ determined by those with $r$th bit 1.
Let $f_0^{(t)},f_1^{(t)}\in\Z_{\geq 0}^n$ be the frequency vectors of the active items in each category, respectively, initialized to 0 at the beginning of the round.
Then the Bernoulli processes are $X_0^{(t)} = \langle Z,f_0^{(t)}\rangle$ and $X_1^{(t)} = \langle Z,f_1^{(t)}\rangle$, where $Z$ is a vector of $4$-wise independent Rademacher random variables (i.e.\ the $Z_i$ are marginally uniformly random in $\{-1,1\}$).

The $r$th round ends when $|X_0+X_1|>\supbound{\sigma}{r-1}$, for specified\footnote{$c=1/32$ and $\beta=3/4$ would suffice.} constants $c$ and $\beta$.
At this point, the algorithm compares the values $|X_0|$ and $|X_1|$ and records the identity of the larger one as the $r$th bit of the candidate heavy hitter.
All those items with $r$th bit corresponding to the smaller counter are discarded (made inactive), and the next round is started.

After $R$ rounds are completed, if there is a heavy hitter then its randomized label will be known with good probability.
The identity of the item can be determined by selecting an item in the stream that passes all of the $R$ bit-wise tests, or by inverting the hash function used for the label.
If it is a $K$-heavy hitter, for a sufficiently large $K=O(1)$, then the algorithm will find it with probability at least $2/3$.
The algorithm is formally presented in Algorithm~\ref{algo: HH}.

The most important technical component of the analysis is the following theorem, which is proved in Section~\ref{sec: chaining}.
Theorem~\ref{thm: sup bound} gives us control of the evolution of $|X_0|$ and $|X_1|$ so we can be sure that the larger of the two identifies a bit of $H$.
\newcommand{\chainingtheoremtext}{If $Z\in \{-1,1\}^{n}$ is drawn from a $4$-wise independent family, $\E \sup_{t}|\langle f^{(t)}, Z\rangle|< 23\cdot \|f^{(m)}\|_2$.}
\begin{theorem}\label{thm: sup bound}
\chainingtheoremtext
\end{theorem}
We will use $C_*<23$ to denote the optimal constant in Theorem~\ref{thm: sup bound}.

The key idea behind the algorithm is that as we learn bits of the heavy hitter and discard other items, it becomes easier to learn additional bits of the heavy hitter's identity. 
With fewer items in the stream as the algorithm proceeds, the heavy hitter accounts for a larger and larger fraction of the remaining stream as time goes on.
As the heavy hitter gets heavier the discovery of the bits of its identity can be sped up.
When the stream does not contain a heavy hitter this acceleration of the rounds might not happen, though
that is not a problem because when there is no heavy hitter the algorithm is not required to return any output.
Early rounds will each use a constant fraction of the updates to the heavy hitter, but the algorithm will be able to finish all $R=\Theta(\log n)$ rounds because of the speed-up.
The parameter $\beta$ controls the speed-up of the rounds.
Any value of $\beta\in(\frac{1}{2},1)$ can be made to work (possibly with an adjustment to $c$), but the precise value affects the heaviness requirement and the failure probability.

\subsection{Proof of Theorem~\ref{thm: sup bound}}\label{sec: chaining}

Let $Z\in \{-1, 1\}^n$ be random. We are interested in bounding $\E \sup_{t}|\langle f^{(t)}, Z\rangle|$.
It was shown in \cite{BravermanCIW16} that if each entry in $Z$ is drawn independently and uniformly from $\{-1,1\}$, then $\E \sup_{t}|\langle f^{(t)}, Z\rangle|\lesssim \|f^{(m)}\|_2$. 
We show that this inequality still holds if the entries of $Z$ are drawn from a $4$-wise independent family, which is used both in our analyses of \hho and our $F_2$ tracking algorithm. 
The following lemma is implied by \cite{Haagerup82}.
\begin{lemma}[Khintchine's inequality]
Let $Z\in\{-1,1\}^n$ be chosen uniformly at random, and $x\in\mathbb{R}^n$ a fixed vector. Then for any even integer $p$, $\E\inprod{Z, x}^p \le \sqrt{p}^p \cdot \|x\|_2^p$.
\end{lemma}
\begin{proof}[of Theorem~\ref{thm: sup bound}]
To simplify notation, we first normalize the vectors in $\{f^{(0)}=0, f^{(1)}, \ldots, f^{(m)}\}$ (i.e., divide by $\|f^{(m)}\|_2$). Denote the set of these normalized vectors by $T=\{v_0,\ldots,v_m\}$, where $\|v_m\|_2=1$. Recall that an $\eps$-net of some set of points $T$ under some metric $d$ is a set of point $T'$ such that for each $t\in T$, there exists some $t'\in T'$ such that $d(t, t') \le \eps$. For every $k\in \mathbb{N}$, we can find a $1/2^k$-net of $T$ in $\ell_2$ with size $|S_k|\le 2^{2k}$ by a greedy construction as follows.

To construct an $\varepsilon$-net for $T$, we first take $v_0$, then choose the smallest $i$ such that $\|v_i-v_0\|_2> \varepsilon$, and so on. To prove the number of elements selected is upper bounded by $1/\varepsilon^2$, let $u_0, u_1, u_2, \ldots, u_t$ denote the vectors we selected accordingly, and note that the second moments of $u_1-u_0, u_2-u_1, \ldots, u_t-u_{t-1}$ are greater than $\varepsilon^2$. Because the vectors $u_{i}-u_{i-1}$ have non-negative coordinates, $\|u_t\|_2^2$ is lower bounded by the summation of these moments, while on the other hand $\|u_t\|_2^2 \le 1$. Hence the net is of size at most $1/\varepsilon^2$.

Let $S$ be a set of vectors. Let $Z\in \{-1,1\}^n$ be drawn from a $p$-wise independent family, where $p$ is an even integer. By Markov and Khintchine's inequality,
\begin{align*}
\Pr(|\langle x,Z\rangle|>\lambda \cdot |S|^{1/p} \cdot \|x\|_2)&< \frac{\E|\langle x,Z \rangle|^p}{\lambda^p\cdot |S|\cdot \|x\|_2^p}\\
 &< \frac 1{|S|}\cdot \left(\frac{\sqrt{p}}{\lambda}\right)^p .
\end{align*}
Hence,
\begin{align*}
\E\sup_{x\in S} &|\langle x,Z\rangle | = \int_0^\infty\Pr(\sup_{x\in S} |\langle x,Z\rangle| > u)du\\
{}&=|S|^{1/p}\cdot \sup_{x\in S}\|x\|_2\cdot \\
&\qquad  \int_0^\infty \Pr(\sup_{x\in S}|\langle x,Z\rangle| > \lambda\cdot |S|^{1/p}\cdot \sup_{x\in S}\|x\|_2) d\lambda\\
{}&< |S|^{1/p}\cdot \sup_{x\in S}\|x\|_2 \cdot \left(\sqrt{p} + \int_{\sqrt{p}}^\infty \left(\frac {\sqrt{p}}{\lambda}\right)^pd\lambda\right)\\
&\qquad \text{ (union bound)}\\
{}&= |S|^{1/p}\cdot \sup_{x\in S}\|x\|_2 \cdot \sqrt{p}\cdot \left(1 + \frac{1}{p-1}\right) 
\end{align*}

Now we apply a similar chaining argument as in the proof of Dudley's inequality (cf.\ \cite{Dudley67}). For $x\in T$, let $x^k$ denote the closest point to $x$ in $S_k$. Then $\|x^k-x^{k-1}\|_2\le \|x^k-x\|_2+\|x-x^{k-1}\|_2\le (1/2^k)+(1/2^{k-1})$.  Note that if for some $x\in T$ one has that $x_k = v_t$ is the closest vector to $x$ in $T_k$ (under $\ell_2$), then the closest vector $x_{k-1}$ to $x$ in $T_{k-1}$ must either be the frequency vector $v_{t'}$ in $T_{k-1}$ such that $t'$ is the smallest timestamp  after $t$ of a vector in $T_{k-1}$, or the largest timestamp before $t$ in $T_{k-1}$. Thus the size of $\{x^k-x^{k-1}|x\in T\}$ is upper bounded by $2|S_k|\le 2^{2k+1}$, implying for $p = 4$ 
\begin{align*}
\E &\sup_{x\in T} |\langle x,Z \rangle| 
\le \sum_{k=1}^{\infty} {\E \sup |\langle x^k-x^{k-1}, Z \rangle|}\\ 
&< 3\cdot 2^{1/p}\sqrt{p}\left(1 + \frac{1}{p-1}\right) \sum_{k=1}^{\infty} {(2^{2k})^{1/p}\cdot (1/2^k)} \\
&< 23.
\end{align*}
\end{proof}

\subsection{Identifying a single heavy hitter given an approximation to $F_2$}\label{sec: HH1}

This section analyzes the subroutine \hho, which is formally presented in Algorithm~\ref{algo: HH}.
The goal of this section is to prove Lemma~\ref{lem: hho correctness}, the correctness of \hho.
We use $H\in[n]$ to stand for the identity of the most frequent item in the stream. 
It is not assumed to be a heavy hitter unless explicitly stated.

\subsubsection{Randomizing the labels}\label{sec: random labels}
The first step of \hho is to choose a hash function $h:[n]\to\{0,1\}^R$, for $R=O(\log n)$, that relabels the universe of items~$[n]$.
For each $r\geq 0$, let
\[\calH_r \coloneqq \{i\in[n]\setminus\{H\}\mid h(i)_k=h(H)_k\text{ for all }1\leq k\leq r\},\]
 and let $\bar{\calH}_r \coloneqq \calH_{r-1}\setminus\calH_r$, with $\bar{\calH}_0=\emptyset$ for convenience.
By definition, $\calH_R\subseteq\calH_{R-1}\subseteq\cdots\subseteq\calH_0=[n]\setminus\{H\}$, and, in round $r\in[R]$, our hope is that the active items are those in $\calH_{r-1}$. 

The point of randomizing the labels is as a sort of ``load balancing'' among the item labels.
The idea is that each bit of $h(H)$ partitions the active items into two roughly equal sized parts, i.e.\, $|\calH_r|\approx|\bar{\calH}_r|$ in every round $r$. 
This leads \hho to discard roughly half of the active items after each round, allowing us to make progress on learning the (hashed) identity of a heavy hitter. 
We will make use of the randomized labels in the next section, within the proof of Lemma~\ref{lem: one round sup}.

For $h$, we recommend choosing a prime $p\approx \min\{n,F_2\}^2$ and assigning the labels $h(i) = a_0+a_1i\bmod{p}$, for $a_0$ and $a_1$ randomly chosen in $\{0,1,\ldots,p-1\}$ and $a_1\neq0$.  
We can always achieve this with $R=3\log_2 (\min\{n,F_2\} +1)$, which is convenient for the upcoming analysis. This distribution on $h$ is known to be a $2$-wise independent family \cite{CarterW79}. Note computing $h(i)$ for any $i$ takes $O(1)$ time. It is also simple to invert: namely $x=a_1^{-1}(h(x)-a_0)\bmod{p}$, so $x$ can be computed quickly from $h(x)$ when $p>n$.  
Inverting requires computing the inverse of $a_1$ modulo $p$, which takes $O(\log \min\{n,F_2\})$ time via repeated squaring, however this computation can be done once, for example during initialization of the algorithm, and the result stored for all subsequent queries. 
Thus, the time to compute $a_1^{-1}\bmod{p}$ is negligible in comparison to reading the stream. 

\subsubsection{Learning the bits of the randomized label}

\begin{figure}
\begin{center}
\begin{tikzpicture}[every node/.style={draw,shape=circle,fill=black,minimum size=6pt,inner sep=0pt}]
\node[label=$\calH_0$] (a0) at (0,0) {};
\node[label=$\calH_1$] (b0) at (1.25,0.5) {};
\node[label=below:{$\bar{\calH}_1$}] (b1) at (1.25,-0.5) {};
\node[label=$\bar{\calH}_2$] (c0) at (2.5,1) {};
\node[label=below:{$\calH_2$}] (c1) at (2.5,0) {};
\node[label=$\calH_{\tiny R-1}$] (d0) at (4.5,0) {};
\node[label=right:{$\calH_{R}$}] (e0) at (5.75,0.5) {};
\node[label=right:{$\bar{\calH}_R$}] (e1) at (5.75,-0.5) {};
\draw (a0) -- node[inner sep=1pt,above,draw=none,fill=none] {0} ++ (b0);
\draw (a0) -- node[inner sep=1pt,below,draw=none,fill=none] {1} ++ (b1);
\draw (b0) -- node[inner sep=1pt,above,draw=none,fill=none] {0} ++ (c0);
\draw (b0) -- node[inner sep=1pt,below,draw=none,fill=none] {1} ++ (c1);
\draw[dashed] (c1) -- (3.125,0.25) node [inner sep=1pt,near end,above,draw=none,fill=none] {0};
\draw[dashed] (c1) -- (3.125,-0.25) node [inner sep=1pt,near end,below,draw=none,fill=none] {1};
\draw[dashed] (d0) -- (3.875,-0.25) node [inner sep=1pt,near end,above,draw=none,fill=none] {0};
\draw (d0) -- node[inner sep=1pt,above,draw=none,fill=none] {0} ++ (e0);
\draw (d0) -- node[inner sep=1pt,below,draw=none,fill=none] {1} ++ (e1);
\end{tikzpicture}
\end{center}
\caption{In this example of the execution of \hho, the randomized label $h(H)$ of the heavy hitter $H$ begins with $01$ and ends with $00$.
Each node in the tree corresponds to a round of \hho, which must follow the path from $\calH_0$ to $\calH_R$ for the output to be correct.}\label{fig: tree}
\end{figure}
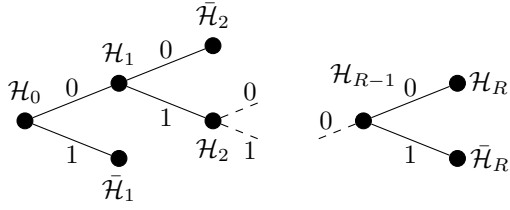

After randomizing the labels \hho proceeds with the series of $R$ rounds to identify the (randomized) label $h(H)$ of the heavy hitter $H$.
The sequence of rounds is depicted in Figure~\ref{fig: tree}.
Each node in the tree corresponds to one round of \hho.
The algorithm traverses the tree from left to right as the rounds progress.  
Correctness of \hho means it traverses the path from $\calH_0$ to $\calH_R$.
The 0/1 labels on the path leading to $\calH_R$ are the bits of $h(H)$, and when $R=O(\log n)$ is sufficiently large we get $\calH_R=\{H\}$ with high probability.
Thus the algorithm correctly identifies $h(H)$, from which it can determine $H$ with the method discussed in Section~\ref{sec: random labels}.

Now let us focus on one round and suppose that $H$ is a $K$-heavy hitter, for some large constant $K$.
Suppose the algorithm is in round $r\geq 1$, and recall that the goal of the round is to learn the $r$th bit of $h(H)$.
Our hope is that the active items are those in $\calH_{r-1}$ (otherwise the algorithm will fail), which means that the algorithm has correctly discovered the first $r-1$ bits of $h(H)$.
The general idea is that \hho partitions $\calH_{r-1}\cup\{H\}$ into $\calH_r\cup\{H\}$ and $\bar{\calH}_r$, creates a Bernoulli process for each of those sets of items, and compares the values of the two Bernoulli processes to discern $h(H)_r$.
Suppose that the active items are indeed $\calH_{r-1}$ and, for the sake of discussion, that the $r$th bit of the heavy hitter's label is $h(H)_r=0$.
Then the Bernoulli processes $X_0$ and $X_1$, defined in Algorithm~\ref{algo: HH}, have the following form
\[X_0(t) = Z_Hf_H^{(s:t)} + \sum_{i\in \calH_r}Z_if_i^{(s:t)},\quad X_1(t) =\sum_{i\in \bar{\calH}_r}Z_if_i^{(s:t)},\]
where $s\leq m$ is the time of the last update to round $r-1$ and $t$ is the current time.
To simplify things a little bit we adopt the notation $f(S)$ for the frequency vector restricted to only items in $S$.
For example, in the equations above become $X_0 = Z_Hf_H^{(s:t)}+\langle Z,f^{(s:t)}(\calH_r)\rangle$ and $X_1 = \langle Z,f^{(s:t)}(\bar{\calH}_r)\rangle$.

The round is a success if $|X_0|>|X_1|$ (because we assumed $h(H)_r=0$) at the first time when $|X_0+X_1|>c\sigma\beta^r$.
When that threshold is crossed, we must have $|X_0|\geq c\sigma\beta^r/2$, $|X_1|\geq c\sigma\beta^r/2$, or both.
The way we will ensure that the round is a success is by establishing the following bound on the Bernoulli process $X_1$:  $|X_1| = |\langle Z,f^{(s:t)}(\bar{\calH}_r)\rangle| < c\sigma\beta^r/2$, for all times $t\geq s$. 
Of course, the round does not end until the threshold is crossed, so we will also establish a bound on the complementary Bernoulli process $|X_0-Z_Hf_H^{(s:t)}| = |\langle Z,f^{(s:t)}(\bar{\calH}_r)\rangle| < c\sigma\beta^r/2$, at all times $t\geq s$.
When this holds we must have $|X_0+X_1|>c\sigma\beta^r$ no later than the first time $t$ where $f_H^{(s:t)}\geq 2c\sigma\beta^r$, so the round ends after at most $2c\sigma\beta^r$ updates to $H$.
In total over all of the rounds this uses up  no more than $\sum_{r\geq0}2c\sigma\beta^r<f_H$ updates to $H$, where we have used $\sigma<2\sqrt{F_2}<3f_H$ by our assumption that $H$ is a heavy hitter. 
In truth, both of those inequalities fail to hold with some probability, but the failure probability is $O(1/\beta^r2^{r/2})$ so the probability that the algorithm succeeds will turn out to be $1-\sum_{r\geq0}O(1/\beta^r2^{r/2})>2/3$.

The next lemma establishes the control on the Bernoulli processes that we have just described (compare the events $E_i$ with the previous paragraph).  
We will use it later with $K\approx\beta^r$ so that, while the rounds progress, the upper bounds on the process maxima and the failure probabilities both decrease geometrically as desired.
This means that the lengths of the rounds decreases geometrically and the latter means that a union bound suffices to guarantee that all of the events $E_i$ occur.
In our notation $F_2-f_H^2=F_2(\calH_0)$.

\begin{lemma}\label{lem: one round sup}
For any $r\in\{0,1,\ldots,R\}$ and $K>0$, the events 
\[E_{2r-1}\coloneqq\left\{\max_{s,t\leq m} |\left\langle Z, f^{(s:t)}(\bar{\calH}_{r})\right\rangle| \leq K F_2(\calH_0)^{1/2}\right\}\] and 
\[E_{2r}\coloneqq\left\{\max_{s,t\leq m} |\left\langle Z, f^{(s:t)}(\calH_r)\right\rangle| \leq K F_2(\calH_0)^{1/2}\right\}\]
have respective probabilities at least $1-\frac{4C_* }{K 2^{r/2}}$ of occurring, where $C_* < 23$ is the constant from Theorem~\ref{thm: sup bound}.
\end{lemma}
\begin{proof}
By the Law of Total Probability and Theorem~\ref{thm: sup bound} with Markov's Inequality we have
\begin{align*}
\Pr&\left(
\max_{t\leq m}|\langle Z,f^{(t)}(\calH_r)\rangle| \geq \frac{1}{2}K F_2(\calH_0)^{1/2}\right)\\
&= \E\left\{\Pr\left(
\max_{t\leq m}|\langle Z,f^{(t)}(\calH_r)\rangle| \geq \frac{1}{2}K F_2(\calH_0)^{1/2}\middle|\calH_r\right)\right\}\\
&\leq \E\left\{\frac{2C_* F_2(\calH_r)^{1/2}}{K F_2(\calH_0)^{1/2}} \right\}\leq \frac{2C_*F_2(\calH_0)^{1/2}}{K F_2(\calH_0)^{1/2}2^{r/2}},
\end{align*}
where the last inequality is Jensen's.
The same holds if $\calH_r$ is replaced by $\bar{\calH}_r$.

Applying the triangle inequality to get $|\langle Z, f^{(s:t)}(\calH_r)\rangle| \leq |\langle Z, f^{(s)}(\calH_r)\rangle| + |\langle Z, f^{(t)}(\calH_r)\rangle|$ we then find
$P(E_{2r})\geq 1-\frac{4C_*}{K2^{r/2}}$.
A similar argument proves $P(E_{2r-1})\geq 1-\frac{4 C_*}{K2^{r/2}}$.
\end{proof}

From here the strategy to prove the correctness of \hho is to  inductively use Lemma~\ref{lem: one round sup} to bound the success of each round.
The correctness of \hho, Lemma~\ref{lem: hho correctness}, follows directly from Lemma~\ref{lem: R rounds}.

Let $U$ be the event $\{h(j)\neq h(H)\text{ for all }j\neq H, f_j>0\}$ which has, by pairwise independence, probability $\Pr(U) \geq 1-F_02^{-R}\geq 1-\frac{1}{\min\{n,F_2\}^2}$, recalling that $F_0\leq\min\{n,F_2\}$ is the number of distinct items appearing before time $m$.
The next lemma is the main proof of correctness for our algorithm.
\begin{lemma}\label{lem: R rounds}
Let $K'\geq 128$, $c=1/32$, and $\beta=3/4$.  
If $2K'C_*\sqrt{F_2(\calH_0)}\leq\sigma\leq 2\sqrt{2}f_H$ and $f_H> 2K'C_*\sqrt{F_2(\calH_0)}$ then, with probability at least $1-\frac{1}{\min\{F_2,n\}^2}-\frac{8}{K'c(\sqrt{2}\beta-1)}$ the algorithm \hho returns $H$.
\end{lemma}
\begin{proof}
Recall that $H$ is \emph{active} during round $r$ if it happens that $h(H)_i=b_i$, for all $1\leq i\leq r-1$, which implies that updates from $H$ are not discarded by the algorithm during round $r$.
Let $K=K(r)=K'cC_*\beta^{r}$ in Lemma~\ref{lem: one round sup}, and let $E$ be the event that $U$ and $\cap_{r=1}^{2R}E_r$ both occur.
We prove by induction on $r$ that if $E$ occurs then either $b_r = h(H)_r$, for all $r\in[R]$ or $H$ is the only item appearing in the stream.  
In either case, the algorithm correctly outputs $H$, where in the former case it follows because $E\subseteq U$. 

Let $r\geq 1$ be such that $H$ is still active in round $r$, i.e. $b_i=h(H)_i$ for all $1\leq i\leq r-1$.
Note that all items are active in round 1.
Since $H$ is active, the remaining active items are exactly $\calH_{r-1}=\calH_{r}\cup\bar{\calH}_{r}$.
Let $t_{r}$ denote the time of the last update received during the $r$th round, and define $t_0=0$.
At time $t_{r-1}\leq t<t_{r}$ we have 
\begin{align*}
c\sigma\beta^{r} &>|X_0+X_1|\\
&= |\langle Z, f^{(t_{r-1}:t)}(\calH_{r}\cup\bar{\calH}_{r})\rangle + Z_Hf^{(t_{r-1}:t)}_H|\\
&\geq f_H^{(t_{r-1}:t)} - K(r-1) F_2(\calH_0)^{1/2},
\end{align*}
where the last inequality follows from the definition of $E_{2(r-1)}$.
Rearranging and using the assumed lower bound on $\sigma$, we get the bound 
\begin{equation}\label{eq: KF2 calc}
K(r-1)F_2(\calH_0)^{1/2}\leq\frac{K(r-1)}{2K'C_*}\sigma = \frac{1}{2}c\sigma\beta^{r-1}.
\end{equation}
Therefore, by rearranging we see $f_H^{(t_{r-1}:t_{r})}\leq 1+f^{(t_{r-1}:t)}_H< 1+\frac{3}{2}c\sigma\beta^{r-1}$.
That implies
$f^{(t_{r})}_H =\sum_{k=1}^{r} f^{(t_{k-1}:t_{k})}_H < r+\frac{3}{2}c\sigma\sum_{k=1}^{r}\beta^{k-1}\leq r+\frac{3\sqrt{2}c}{1-\beta}f_H$.
Thus, if $f_H- \frac{3\sqrt{2}c}{1-\beta}f_H>R$ then round $r\leq R$ is guaranteed to be completed and a further update to $H$ appears after the round.
Suppose, that is not the case, and rather $R\geq f_H-\frac{3\sqrt{2}c}{1-\beta}f_H\geq\frac{1}{2}f_H$, where the last inequality follows from our choices $\beta=3/4$ and $c=1/32$.
Then, by the definition of $R$,  $9(1+\log_28f_H^2)^2\geq R^2\geq \frac{1}{4} f_H^2$.  
One can check that this inequality implies that $f_H\leq 104$.
Now $K'\geq 128$ and the heaviness requirement of $H$ implies that $F_2(\calH_0)=0$.
Therefore, $H$ is the only item in the stream, and, in that case the algorithm will always correctly output $H$.

Furthermore, at the end of round $r$, $|X_0+X_1|\geq c\sigma\beta^{r}$, so we have must have either $|X_0|\geq c\sigma\beta^{r}/2$ or $|X_1|\geq c\sigma\beta^{r}/2$. 
Both cannot occur for the following reason.
The events $E_{2r-1}$ and $E_{2r}$ occur, recall these govern the non-heavy items contributions to $X_0$ and $X_1$, and these events, with the inequality \eqref{eq: KF2 calc}, imply 
\[|\langle Z,f^{(t_{r-1}:t_{r})}(\calH_r)\rangle|\leq K(r)F_2(\calH_0)^{1/2}<\frac{1}{2}c\sigma\beta^r\]
and the same holds for $\bar{\calH_r}$. 
Therefore, the Bernoulli process not including $H$ has its value smaller than $c\sigma\beta^r/2$, and the other, larger process identifies the bit $h(H)_r$.
By induction, the algorithm completes every round $r=1,2,\ldots,R$ and there is at least one update to $H$ after round $R$.
This proves the correctness of the algorithm assuming the event $E$ occurs.

It remains to compute the probability of $E$.
Lemma~\ref{lem: one round sup} provides the bound
\begin{align*}
\Pr(U\text{ and }\cap_{i=1}^{2R}E_i ) 
&\geq 1 - \frac{1}{\min\{n,F_2\}^2} - \sum_{r=0}^{R}\frac{8C_*}{K(r)2^{r/2}}\\
&= 1 - \frac{1}{\min\{n,F_2\}^2} - \sum_{r=0}^{R}\frac{8}{K'c\beta^{r}2^{r/2}}\\
&> 1-\frac{1}{\min\{n,F_2\}^2}-\frac{8}{K'c(\sqrt{2}\beta-1)}.
\end{align*}
\end{proof}

\begin{proposition}\label{lem: sigma to f_H}
Let $\alpha\geq 1$.
If $F_2^{1/2}\leq\sigma\leq 2F_2^{1/2}$ and $f_H\geq \alpha \sqrt{F_2(\calH_0)}$ then $\alpha\sqrt{F_2(\calH_0)}\leq \sigma\leq 2\sqrt{2}f_H.$
\end{proposition}
\begin{proof}
$\sigma^2 \geq F_2\geq f_H^2 =F_2-F_2(\calH_0)\geq (1-\frac{1}{1+\alpha^2})F_2\geq\frac{1}{8}\sigma^2$.
\end{proof}

\begin{lemma}[\hho Correctness]\label{lem: hho correctness}
There is a constant $K$ such that if $H$ is a $K$-heavy hitter and $\sqrt{F_2}\leq\sigma\leq2\sqrt{F_2}$, then with probability at least $2/3$ algorithm \hho returns $H$.
\hho uses $O(1)$ words of storage.
\end{lemma}
\begin{proof}
The Lemma follows immediately from Proposition~\ref{lem: sigma to f_H} and Lemma~\ref{lem: R rounds} by setting $K'=2^{13}$, which allows $K=2^{14}C_*\leq 380,000$.
\end{proof}

\subsection{$F_2$ Tracking}\label{sec: ams sup}
This section proves that the AMS algorithm with 8-wise, rather than 4-wise, independent random signs has an additive $\epsilon F_2$ approximation guarantee at all points in the stream.
We will use the tracking to ``guess'' a good value of $\sigma$ for input to \hho, but, because the AMS algorithm is a fundamental streaming primitive, it is of independent interest from the \bptree algorithm.
The following theorem is a direct consequence of Lemma~\ref{lem: F2 tracker chaining} and~\cite{AlonMS99}.
\begin{theorem}\label{thm: beefed up AMS}
Let $0<\epsilon<1$.
There is a streaming algorithm that outputs at each time $t$ a value $\hat{F}_2^{(t)}$ such that
$\Pr(|\hat{F}_2^{(t)}-F_2^{(t)}|\leq \epsilon F_2,\text{ for all }0\leq t\leq m)\geq 1-\delta$.
The algorithm use $O(\frac{1}{\epsilon^2}\log\frac{1}{\delta})$ words of storage and has $O(\frac{1}{\epsilon^{2}}\log\frac{1}{\delta})$ update time.
\end{theorem}

Let us remark that it follows from Theorem~\ref{thm: beefed up AMS} and a union bound that one can achieve a $(1\pm\epsilon)$ \emph{multiplicative} approximation to $F_2$ at all points in the stream using $O(\epsilon^{-2}\log\log m)$ words.
The proof that this works breaks the stream into $O(\log{m})$ intervals of where the change in $F_2$ doubles.  

In its original form~\cite{AlonMS99}, the AMS sketch is a product of the form $\Pi f$, where $\Pi$ is a random $k\times n$ matrix chosen with each row independently composed of 4-wise independent $\pm1$ random variables.
The sketch uses $k=\Theta(1/\epsilon^2)$ rows to achieve a $(1\pm\epsilon)$-approximation with constant probability.
We show that the AMS sketch with $k \simeq 1/\epsilon^2$ rows and $8$-wise independent entries provides $\ell_2$-tracking with additive error $\epsilon \|f\|_2$ at all times.
We define $v_t = f^{(t)} / \|f^{(m)}\|_2$ so $\|v_t\|_2 \le 1$, for all $t\geq0$, and $\|v_m\|_2 = 1$. Define $T = \{v_0,v_1,\ldots,v_m\}$. 
We use $\|A\|$ to denote the spectral norm of $A$, which is equal to its largest singular value, and $\|A\|_F$ for the Frobenious norm, which is the Euclidean length of $A$ when viewed as a vector.
Our proof makes use of the following moment bound for quadratic forms.
Recall that given a metric space $(X,d)$ and $\epsilon>0$, an \emph{$\epsilon$-net} of $X$ is a subset $N\subseteq X$ such that $d(x,N)=\inf_{y\in N}d(x,y)\leq\epsilon$ for all $x\in X$.

\begin{theorem}[Hanson-Wright~\cite{hanson1971bound}]
For $B\in\mathbb{R}^{n\times n}$ symmetric with $(Z_i)$ uniformly random in $\{-1,1\}^n$, for all $p\ge 1$,
$
\|Z^T B Z - \E Z^T B Z\|_p \lesssim \sqrt{p}\|B\|_F + p\|B\| .
$
\end{theorem}

Observe the sketch can be written $\Pi x = A_{x} Z$, where
\begin{align*}
\label{eqn:AdelxDef}
A_x &:= \frac{1}{\sqrt{k}}\sum_{i=1}^k\sum_{j=1}^n x^i_j e_i\otimes e_{n(i-1)+j} \\
&= \frac{1}{\sqrt{k}}\begin{bmatrix} 
- x - & 0 & \cdots & 0\\
0 & - x - & \cdots & 0\\
\vdots &\vdots &  &\vdots\\
0&0&\cdots& - x -
\end{bmatrix} .
\end{align*}
We are thus interested in bounding
$\E_{Z} \sup_{x\in T} | Z^T B_x Z - \E Z^T B_x Z  |,  $
for $B_x = A_x^T A_x$. 
Note for any $\|x\|_2, {\|y\|_2 \le 1}$,
\begin{equation}
\|xx^T - yy^T\|_F \le 4\|x-y\|_2 \label{eqn:frob}.
\end{equation}

\begin{lemma}[$F_2$ tracking]\label{lem: F2 tracker chaining}
If $k\gtrsim 1/\epsilon^2$ and $Z\in\{-1,+1\}^{kn}$ are $8$-wise independent then 
\[\E\sup_{t}\left|\|\Pi f^{(t)}\|_2^2 - \|f^{(t)}\|_2^2\right|\leq \epsilon\|f\|_2^2. \]
\end{lemma}
\begin{proof}
Let $A_x$, $x\in T$, as defined above and $B_x=A_x^TA_x$.
By \eqref{eqn:frob}, $\|B_x-B_y\|_F\leq \frac{4}{\sqrt{k}}\|x-y\|_2$, for all $x,y\in T$.
In particular, $\sup_{x\in T} \|B_x\| \leq \sup_{x\in T} \|B_x\|_F \leq 1/\sqrt{k}$. 
Let $T_\ell$ be a $(1/2^\ell)$-net of $T$ under $\ell_2$; we know we can take $|T_\ell| \le 4^\ell$. 
$\mathcal{B}_\ell = \{ B_x : x \in T_\ell\}$ is a $1/\sqrt{k}2^\ell$-net under $\|\cdot\|$ and also under $\|\cdot\|_F$. 
For $x\in T$, let $x_\ell\in T_\ell$ denote the closest element in $T_\ell$, under $\ell_2$. 
Then we can write
$B_x = B_{x_0} + \sum_{\ell=1}^\infty \Delta_{x_\ell}$,
where $\Delta_{x_\ell}=B_{x_\ell} - B_{x_{\ell-1}}$.
For brevity, we will also define $\gamma(A)\coloneqq |Z^TAZ^T-\E Z^TAZ^T|$.
Thus if the $(Z_i)$ are $2p$-wise independent
\begin{align}
\nonumber \E \sup_{x\in T} \gamma(B_x) &\le \E \sup_{x\in T} \gamma(B_{x_0}) + \E \sup_{x\in T} \sum_{\ell=1}^\infty \gamma(\Delta_{x_\ell})\\
{}&\lesssim \frac{p}{\sqrt{k}} + \sum_{\ell=1}^\infty \E \sup_{x\in T}  \gamma(\Delta_{x_\ell})\label{eqn:chain-b}
\end{align}
If $A\in\mathbb{R}^{n\times n}$ is symmetric, then by the Hanson-Wright Inequality
\begin{align*}
\Pr(&\gamma(A)> \lambda\cdot S^{1/p})< \frac 1S\cdot \left[\left(\frac {C\sqrt{p}\|A\|_F}{\lambda}\right)^p + \left(\frac {Cp\|A\|}{\lambda}\right)^p\right]
\end{align*}
for some constant $C > 0$. Thus if $\mathcal{A}$ is a collection of such matrices, $|\mathcal{A}| = S$, choosing $u^* = C (\sqrt{p} \cdot\sup_{A\in\mathcal{A}} \|A\|_F + p \cdot\sup_{A\in\mathcal{A}} \|A\|)$
\begin{align}
\nonumber \E& \sup_{A\in\mathcal{A}} \gamma(A) = \int_0^\infty \Pr(\sup_{A\in\mathcal{A}} \gamma(A) > u) du\\
\nonumber {}&= S^{1/p} \cdot \int_0^\infty \Pr(\sup_{A\in\mathcal{A}} \gamma(A) > \lambda\cdot S^{1/p}) d\lambda\\
\nonumber {}& = S^{1/p} (u^* + \int_{u^*}^\infty \Pr(\sup_{A\in\mathcal{A}} \gamma(A) > \lambda\cdot S^{1/p}) d\lambda)\\
{}&\lesssim S^{1/p} (\sqrt{p} \cdot\sup_{A\in\mathcal{A}} \|A\|_F + p \cdot\sup_{A\in\mathcal{A}} \|A\|) \label{eqn:chain}
\end{align}

Now by applying \eqref{eqn:chain} to \eqref{eqn:chain-b} repeatedly with $\mathcal{A} = \mathcal{A}_\ell = \{ B_{x_\ell} - B_{x_{\ell-1}}: x \in T\}$, noting $\sup_{A\in\mathcal{A}_\ell} \|A\|\leq\sup_{A\in\mathcal{A}_\ell} \|A\|_F \leq 1/\sqrt{k}2^\ell$ and $|\mathcal{A}_\ell| \le 2|T_\ell| \le 2\cdot 2^{2\ell}$,
\begin{align}
\nonumber
 \E \sup_{x\in T} \gamma(B_x)  \lesssim \frac{p}{\sqrt{k}} + \sum_{\ell=1}^\infty \frac{p2^{2\frac{\ell}{p} - \ell}}{\sqrt{k}}  \lesssim \frac{p}{\sqrt{k}}
\end{align}
for $p\ge 4$. 
Thus it suffices for the entries of $Z$ to be $2p$-wise independent, i.e.\ $8$-wise independent.
\end{proof}

\subsection{The complete heavy hitters algorithm}\label{sec: main algo proof}

We will now describe \hht, formally Algorithm~\ref{algo: guess}, which is an algorithm that removes the assumption on $\sigma$ needed by \hho. 
It is followed by the complete algorithm \bptree.
The step in \hht that guesses an approximation $\sigma$ for $\sqrt{F_2}$ works as follows. 
We construct the estimator $\hat{F}_2$ of the previous section to (approximately) track $F_2$.
\hht starts a new instance of \hho each time the estimate $\hat{F}_2$ crosses a power of 2. 
Each new instance is initialized with the current estimate of $\sqrt{F_2}$ as the value for $\sigma$, but \hht maintains only the two most recent copies of \hho. 
Thus, even though, overall, it may instantiate $\Omega(\log n)$ copies of \hho at most two will running concurrently and the total storage remains $O(1)$ words.
At least one of the thresholds will be the ``right'' one, in the sense that \hho gets initialized with $\sigma$ in the interval $[\sqrt{F_2},2\sqrt{F_2}]$, so we expect the corresponding instance of \hho to identify the heavy hitter, if one exists.

The scheme could fail if $\hat{F}_2$ is wildly inaccurate at some points in the stream, for example if $\hat{F}_2$ ever grows too large then the algorithm could discard every instance of \hho that was correctly initialized. 
But, Theorem~\ref{thm: beefed up AMS} guarantees that it fails only with small probability.

We begin by proving the correctness of \hht in Lemma~\ref{lem: main} and then complete the description and correctness of \bptree in Theorem~\ref{thm: like CountSketch}.

\begin{algorithm}[t]
  \begin{algorithmic}
    \Procedure{\hht}{$p_1$, $p_2$,\,\ldots,$p_m$}
    \State Run $\hat{F}_2$ from Theorem~\ref{thm: beefed up AMS} with $\epsilon=1/100$ and $\delta=1/20$
    \State Start \hho($1$, $p_1$,\ldots, $p_m$)
    \State Let $t_0=1$ and $t_k = \min\{t \mid \hat{F}_2^{(t)}\geq 2^k\}$, for $k\geq 1$.
    \For{each time $t_k$} 
      \State Start $\hho( (\hat{F}_2^{(t_k)})^{1/2}, p_{t_k},p_{t_k+1},\ldots p_m)$
      \State Let $H_k$ denote its output if it completes
      \State Discard $H_{k-2}$ and the copy of \hho started at $t_{k-2}$
    \EndFor
    \State \Return $H_{k-1}$
    \EndProcedure
  \end{algorithmic}
  \caption{Identify a heavy hitter by guessing $\sigma$.}\label{algo: guess}
\end{algorithm}

\newcommand{\maintheoremtext}{There exists a constant $K>0$ and a 1-pass streaming algorithm~\hht, Algorithm~\ref{algo: guess}, such that if the stream contains a $K$-heavy hitter then with probability at least $0.6$ \hht returns the identity of the heavy hitter.
The algorithm uses $O(1)$ words of memory and $O(1)$ update time.}
\begin{lemma}\label{lem: main}
\maintheoremtext
\end{lemma}
\begin{proof}
The space and update time bounds are immediate from the description of the algorithm.
The success probability follows by a union bound over the failure probabilities in Lemma~\ref{lem: hho correctness} and Theorem~\ref{thm: beefed up AMS}, which are $1/3$ and $\delta=0.05$ respectively.
It remains to prove that there is a constant $K$ such that conditionally given the success of the $F_2$ estimator, the hypotheses of Lemma~\ref{lem: hho correctness} are satisfied by the penultimate instance of \hho by \hht.

Let $K'$ denote the constant from Lemma~\ref{lem: hho correctness} and set $K=12K'$, so if $H$ is a $K$-heavy hitter then for any $\alpha>2/K$ and in any interval $(t,t']$ where $(F_2^{(t')})^{1/2}-(F_2^{(t)})^{1/2}\geq \alpha\sqrt{F_2}$ we will have \[f_H^{(t:t')} + F_2(\calH_0)^{1/2}\geq\|f^{(t:t')}\|_2\geq\|f^{(t')}\|_2-\|f^{(t)}\|_2 \geq \alpha\sqrt{F_2}.\]
If follows with in the stream $p_{t},p_{t+1},\ldots,p_{t'}$ the heaviness of $H$ is at least 
\begin{equation}\label{eq: heaviness value}
\frac{f_H^{(t:t')}}{F_2^{(t:t')}(\calH_0)^{1/2}}\geq\frac{\alpha\sqrt{F_2}-\sqrt{F_2(\calH_0)}}{\sqrt{F_2(\calH_0)}}\geq K\alpha-1\geq \frac{K\alpha}{2}.
\end{equation}
Of course, if $F_2^{(t:t')}(\calH_0)=0$ the same heaviness holds.

Let $k$ be the last iteration of \hht.  By the definition of $t_k$, we have $(\hat{F}_2^{(t_{k-1})})^{1/2}\geq \frac{1}{4}(\hat{F}_2)^{1/2}\geq \frac{1}{4}\sqrt{(1-\epsilon)F_2}$.
Similar calculations show that there exists a time $t_*> t_{k-1}$ such that $(F_2^{(t_*)})^{1/2}-(F_2^{(t_{k-1})})^{1/2}\geq \frac{1}{6}F_2$ and $\|f^{t_{k-1}:t_*}\|_2\leq (\hat{F}_2^{(t_{k-1})})^{1/2} \leq 2\|f^{t_{k-1}:t_*}\|_2$. 
In particular, the second pair of inequalities implies that $\hat{F}_2^{(t_k-1)}$ is a good ``guess'' for $\sigma$ on the interval $(t_{k-1},t^*]$.
We claim $H$ is a $K'$ heavy hitter on that interval, too.
Indeed, because of \eqref{eq: heaviness value}, with $\alpha=1/6$, we get that $H$ is a $K'$-heavy hitter on the interval $(t_{k-1},t^*]$.
This proves that the hypotheses of Lemma~\ref{lem: hho correctness} are satisfied for the stream $p_{t_{k-1}+1},\ldots,p_{t^*}$.
It follows that from Lemma~\ref{lem: hho correctness} that \hho correctly identifies $H_{k-1}=H$ on that substream and the remaining updates in the interval $(t_*,t_m]$ do not affect the outcome.
\end{proof}

A now standard reduction from $\epsilon$-heavy hitters to $O(1)$-heavy hitters gives the following theorem.  
The next section describes an implementation that is more efficient in practice.
\newcommand{\likeCountSketchtext}{For any $\epsilon>0$ there is 1-pass streaming algorithm \bptree that, with probability at least $(1-\delta)$, returns a set of $\frac{\epsilon}{2}$-heavy hitters containing every $\epsilon$-heavy hitter and an approximate frequency for every item returned satisfying the $(\epsilon,1/\epsilon^2)$-tail guarantee.
The algorithm uses $O(\frac{1}{\epsilon^2} \log\frac{1}{\epsilon\delta})$ words of space and has $O(\log\frac{1}{\epsilon\delta})$ update and $O(\epsilon^{-2}\log\frac{1}{\epsilon\delta})$ retrieval time.}
\begin{theorem}\label{thm: like CountSketch}
\likeCountSketchtext 
\end{theorem}
\begin{proof}
The algorithm \bptree constructs a hash table in the same manner as \countsketch where the items are hashed into $b = O(1/\epsilon^2)$ buckets for $r=O(\log 1/\epsilon\delta)$ repetitions.
On the stream fed into each bucket we run an independent copy of \hht.
A standard $r\times b$ \countsketch is also constructed.
The constants are chosen so that when an $\epsilon$-heavy hitter in the stream is hashed into a bucket it becomes a $K$-heavy hitter with probability at least 0.95.
Thus, in any bucket with a the $\epsilon$-heavy hitter, the heavy hitter is identified with probability at least 0.55 by Lemma~\ref{lem: main} and the aforementioned hashing success probability.

At the end of the stream, all of the items returned by instances of \hht are collected and their frequencies checked using the \countsketch. 
Any items that cannot be $\epsilon$-heavy hitters are discarded.
The correctness of this algorithm, the bound on its success probability, and the $(\epsilon,1/\epsilon^2)$-tail guarantee follow directly from the correctness of \countsketch and the fact that no more than $O(\eps^{-2}\log(1/\delta\epsilon))$ items are identified as potential heavy hitters. 
\end{proof}

We can amplify the success probability of \hht to any $1-\delta$ by running $O(\log (1/\delta))$ copies in parallel and taking a majority vote for the heavy hitter.
This allows one to track $O(1)$-heavy hitters at all points in the stream with an additional $O(\log\log{m})$ factor in space and update time. 
The reason is because there can be a succession of at most $O(\log m)$ $2$-heavy hitters in the stream, since their frequencies must increase geometrically, so setting $\delta = \Theta(1/\log{m})$ is sufficient.
The same scheme works for \bptree tree, as well, and if one replaces each of the counters in the attached \countsketch with an $F_2$-at-all-times estimator of~\cite{BravermanCIW16} then one can track the frequencies of all $\epsilon$-heavy hitters at all times as well.
The total storage becomes $O(\frac{1}{\epsilon^2}(\log\log n + \log\frac{1}{\epsilon}))$ words and the update time is $O(\log\log n + \log\frac{1}{\epsilon})$.

\section{Experimental Results}\label{sec: experiments}

We implemented \hht in C to evaluate its performance and compare it against the \countsketch for finding one frequent item.
The source code is available from the authors upon request.
In practice, the hashing and repetitions perform predictably, so the most important aspect to understand the performance of \bptree is determine the heaviness constant $K$ where \hht reliably finds $K$-heavy hitters.
Increasing the number of buckets that the algorithm hashes to effectively decreases $n$.
Therefore, in order to maximize the ``effective'' $n$ of the tests that we can perform within a reasonable time, we will just compare \countsketch against \hht.

The first two experiments help to determine some parameters for \hht and the heaviness constant $K$.
Afterwards, we compare the performance of \hht and \countsketch for finding a single heavy hitter in the regime where the heavy hitter frequency is large enough so that both algorithms work reliably.

{\bf Streams.}
The experiments were performed using four types of streams (synthetic data).  
In all cases, one heavy hitter is present.
For a given $n$ and $\alpha$ there are $n$ items with frequency 1 and one item, call it $H$, with frequency $\alpha\sqrt{n}$.
If $\alpha$ is not specified then it is taken to be $1$.
The four types of streams are (1) all occurrences of $H$ occur at the start of the stream, (2) all occurrences of $H$ at the end of the stream, (3) occurrences of $H$ placed randomly in the stream, and (4) occurrences of $H$ placed randomly in blocks of $n^{1/4}$.

The experiments were run on a server with two 2.66GHz Intel Xenon X5650 processors, each with 12MB cache, and 48GB of memory. 
The server was running Scientific Linux 7.2.

\subsection{$F_2$ tracking experiment}

The first experiment tests the accuracy of the $F_2$ tracking for different parameter settings. 
We implemented the $F_2$ tracking in $C$ using the speed-up of \cite{ThorupZ04}, which can be proved to correct for tracking using Appendix~\ref{app: another F2}.
The algorithm uses the same $r\times b$ table as a \countsketch. 
To query $F_2$ one takes the median of $r$ estimates, each of which is the sum of the squares of the $b$ entries in a row of the table.  
The same group of ten type (3) streams with $n=10^8$ and $\alpha=1$ was used for each of the parameter settings.

The results are presented in Table~\ref{tab: F2 tracking}.  
Given the tracker $\hat{F}_2(t)$ and true evolution of the second moment~$F_2(t)$, we measure the maximum $F_2$ tracking error of one instance as $\max_t |\hat{F}_2(t)-F_2(t)|/F_2$, where $F_2$ is the value of the second moment at the end of the stream.
We report the average maximum tracking error and the worst (maximum) maximum tracking over each of the ten streams for every choice of parameters.

The table indicates that, for every choice of parameter settings, the worst maximum tracking error is not much worse than the average maximum tracking error.  
We observe that the tracking error has relatively low variance, even when $r=1$.
It also shows that the smallest possible tracker, with $r=b=1$, is highly unreliable.

\begin{table*}\label{tab: F2 tracking}\def\arraystretch{1.1}
\begin{center}
\begin{tabular}{|r|l|l|l|l|l|l|l|l|l|l|}\hline
 & \multicolumn{5}{|c|}{avg.\ maximum $F_2$ tracking error} & \multicolumn{5}{|c|}{worst maximum $F_2$ tracking error}\\\hline
{\bf b} $\backslash$ {\bf r} & \multicolumn{1}{|r|}{\bf 1} & \multicolumn{1}{|r|}{\bf 2} & \multicolumn{1}{|r|}{\bf 4} & \multicolumn{1}{|r|}{\bf 8} & \multicolumn{1}{|r|}{\bf 16} & \multicolumn{1}{|r|}{\bf 1} & \multicolumn{1}{|r|}{\bf 2} & \multicolumn{1}{|r|}{\bf 4} & \multicolumn{1}{|r|}{\bf 8} & \multicolumn{1}{|r|}{\bf 16} \\\hline
{\bf 1}    & 1.2   & 0.71  & 0.82  & 0.66  & 0.59  & 4.3    & 1.2   & 2.7   & 0.85  & 0.86  \\\hline
{\bf 10}   & 0.35  & 0.30  & 0.33  & 0.19  & 0.16  & 1.1    & 0.68  & 0.91  & 0.28  & 0.20  \\\hline
{\bf 100}  & 0.12  & 0.095 & 0.080 & 0.074 & 0.052 & 0.24   & 0.17  & 0.13  & 0.13  & 0.10  \\\hline
{\bf 1000} & 0.044 & 0.030 & 0.028 & 0.018 & 0.017 & 0.076  & 0.060 & 0.045 & 0.029 & 0.024 \\\hline
\end{tabular}
\end{center}
\vspace{-0.5cm}

\caption{Average and maximum $F_2$ tracking error over 10 streams for different choices of $b$ and $r$.}
\end{table*}

\subsection{Implementations of \hht and \countsketch}\label{sec: impl details}

{\bf \hht implementation details.} 
We have implemented the algorithm \hht as described in Algorithm~\ref{algo: guess}.  
The maximum number of rounds is $R=\min\{\lceil 3\log_2 n\rceil,64\}$.
We implemented the four-wise independent hashing using the ``CW'' trick using the $C$ code from \cite{thorup2012tabulation} Appendix A.14.
We use the code from Appendix A.3 of \cite{thorup2012tabulation} to generate 2-universal random variables for random relabeling of the item.
The $F_2$ tracker from the previous section was used, we found experimentally that setting the tracker parameters as $r=1$ and $b=30$ is accurate enough for \hht.
We also tried four-wise hashing with the tabulation-based hashing for 64-bit keys with 8 bit characters and compression as implemented in $C$ in Appendix A.11 of \cite{thorup2012tabulation}.
This led to a 48\% increase in speed (updates/millisecond), but at the cost of a 55 times increase in space. 

{\bf \countsketch implementation details.}
We implemented \countsketch in C as described in the original paper~\cite{CharikarCF04} with parameters that lead to the smallest possible space. 
We use the \countsketch parameters as $b=2$ (number of buckets/row) and $r=\lceil 3+\log_2n\rceil$ (number of rows).
The choice of $b$ is the smallest possible value. 
The choice of $r$ is the minimum needed to guarantee that, with probability $7/8$, there does not exist an item $i\in[n]\setminus \{H\}$ that collides the heavy hitter in every row of the data structure.
In particular, if we use only $r' < r$ rows then we expect $2^{\log_2 N -r'}$ collisions with the heavy hitter, which would break the guarantee of the \countsketch.
Indeed, suppose there is a collision with the heavy hitter and consider a stream where all occurrences of $H$ appear at the beginning, then \countsketch will not correctly return $H$ as the most frequent item because some item that collides with it and appears after it will replace $H$ as the candidate heavy hitter in the heap.  
In our experiments, the \countsketch does not reliably find the $\alpha$-heavy hitter with these parameters when $\alpha <32$.
This gives some speed and storage advantage to the \countsketch in the comparison against \hht, since $b$ and/or $r$ would need to increase to make \countsketch perform as reliably as \hht during these tests. 

We also tried implementing the four-wise hashing with the Thorup-Zhang tabulation.  
With the same choices of $b$ and $r$ this led to an 18\% speed-up and a 192 times average increase in space.
Since the hash functions are such a large part of the space and time needed by the data structure this could likely be improved by taking $b>2$, e.g.\ $b=100$, and $r\approx\lceil\log_b n\rceil$.
No matter what parameters are chosen the storage will be larger than using the CW trick because each tabulation-based hash function occupies 38kB, which already ten times larger than the whole \countsketch table.

\subsection{Heaviness}

The goal of this experiment is to approximately determine the minimum value $K$ where if $f_H\geq K\sqrt{n}$ then \hht correctly identifies $H$. 
As shown in Lemma~\ref{lem: main}, $K\leq 12\cdot380,000$ but we hope this is a very pessimistic bound.
For this experiment, we take $n=10^8$ and consider $\alpha\in\{1,2,2^2,\ldots,2^6\}$.
For each value of $\alpha$ and all four types of streams we ran \hht one hundred times independently.
Figure~\ref{fig: heaviness} displays the success rate, which is the fraction of the one hundred trials where \hht correctly returned the heavy hitter.
The figure indicates that \hht succeeds reliably when $\alpha\geq 32$.
\begin{figure}\label{fig: heaviness}
\begin{tikzpicture}
\begin{semilogxaxis}[xlabel=$\alpha$, ylabel=success rate,xtick={1,2,4,8,16,32,64},xticklabels={1,2,4,8,16,32,64},legend pos=south east,width=\columnwidth,height=6cm]
\addplot coordinates {(1,0.94) (2,1) (4,1) (8,1) (16,1) (32,1) (64,1)};
\addplot coordinates {(1,.04) (2,.27) (4,1) (8,1) (16,1) (32,1) (64,1)};
\addplot coordinates {(1,0) (2,0) (4,0.02) (8,0.14) (16,0.71) (32,1) (64,1)};
\addplot coordinates {(1,0) (2,0.01) (4,0.02) (8,.25) (16,.78) (32,1) (64,1)};
\legend{(1) start, (2) end, (3) random, (4) blocks}
\end{semilogxaxis}
\end{tikzpicture}
\vspace{-0.5cm}

\caption{Success rate for \hht on four types of streams with $n=10^8$ and heavy hitter frequency $\alpha\sqrt{n}$.}
\label{fig: KN}
\end{figure}
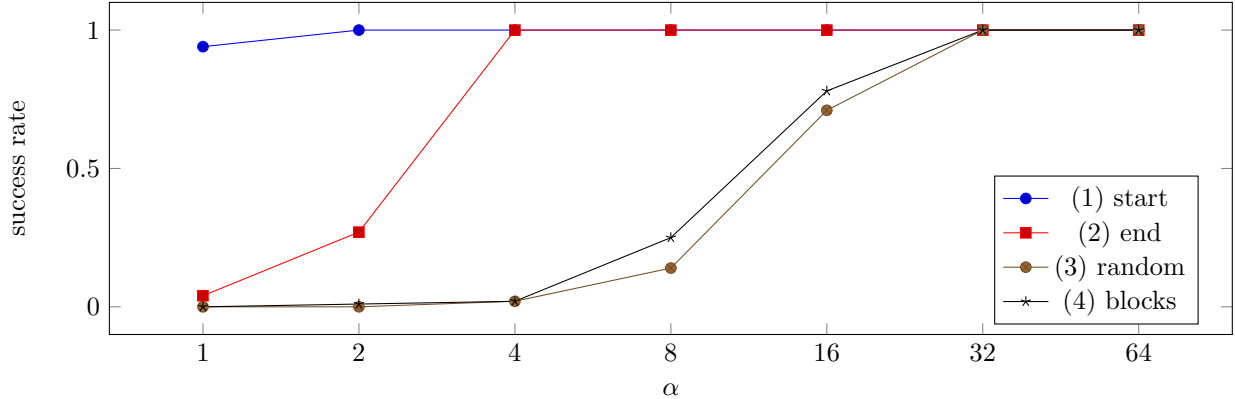

\subsection{\hht versus \countsketch comparison}

In the final experiment we compare \hht against \countsketch.
The goal is to understand space and time trade-off in a regime where both algorithms reliably find the heavy hitter.

For each choice of parameters we compute the update rate of the \countsketch and \hht (in updates/millisecond) and the storage used (in kilobytes) for all of the variables in the associated program. 
The results are presented in Figure~\ref{fig: countsketch}.

The figure shows that \hht is much faster and about one third of the space. 
The dramatic difference in speed is to be expected because two bottlenecks in \countsketch are computing the median of the $\Theta(\log n)$ values and evaluating $\Theta(\log n)$ hash functions.
\hht removes both bottlenecks.
Furthermore, as the subroutine \hho progresses a greater number of items are rejected from the stream, which means the program avoids the associated hash function evaluations in \hht. 
This phenomena is responsible for the observed \emph{increase} in the update rate of \hht as $n$ increases. 
An additional factor that contributes to the speedup is amortization of the start-up time for \hht and of the restart time for each copy of \hho.
\begin{figure}
\begin{center}
\begin{tikzpicture}[scale=0.7]
\begin{axis}[axis y line=left, 
xlabel=$n$, ylabel={\color{blue}rate (updates/ms)}, xtick={6,7,8,9}, xticklabels={$10^6$,$10^7$,$10^8$,$10^{9}$}, ymin=0, ymax=6000,width=1.1*\columnwidth,height=6cm]
\addplot[mark=square*,thick,blue] coordinates {(6,274) (7,245) (8,227) (9,204)}; 
\addplot[mark=*,thick,blue] coordinates {(6,4237) (7,4480) (8,4966) (9,5204) }; 
\end{axis}
\begin{axis}[axis y line=right, axis x line=none, 
xlabel=$n$, ylabel={\color{red}space (kB)}, xtick={6,7,8,9}, xticklabels={$10^6$,$10^7$,$10^8$,$10^{9}$}, ymin=0, ymax=5.5,width=1.1*\columnwidth,height=6cm]
\addplot[mark=square,thick,red] coordinates {(6,3.2) (7,3.8) (8,4.1) (9,4.5)}; 
\addplot[mark=o,thick,red] coordinates {(6,1.4) (7,1.4) (8,1.4) (9,1.4)};
\end{axis}
\end{tikzpicture}
\end{center}
\vspace{-0.75cm}

\caption{Update rate in updates/ms ({\color{blue}$\bullet$}) and storage in kB ({\color{red}$\circ$}) for \hht  and \countsketch ({\color{blue}$\scriptscriptstyle\blacksquare$} and {\color{red}$\scriptscriptstyle\square$}, respectively) with the CW trick hashing.
}\label{fig: countsketch}
\end{figure}

\subsection{Experiments summary}
We found \hht to be faster and smaller than \countsketch. 
The number of rows strongly affects the running time of \countsketch because during each update $r$ four-wise independent hash functions must be evaluated and of a median $r$ values is computed.
The discussion in Section~\ref{sec: impl details} explains why the number of rows $r$ cannot be reduced by much without abandoning the \countsketch guarantee or increasing the space.
Thus, when there is a $K$-heavy hitter for sufficiently large~$K$ our algorithm significantly outperforms \countsketch.
Experimentally we found $K=32$ was large enough.

The full \bptree data structure is needed to find an item with smaller frequency, but for finding an item of smaller frequency \countsketch could outperform \bptree until $n$ is very large. 
For example, to identify an $\alpha$-heavy hitter in the stream our experiments suggest that one can use a \bptree structure with about $\lceil(32/\alpha)^2\rceil$ buckets per row.  
In comparison a \countsketch with roughly $\max\{2,1/\alpha^2\}$ buckets per row should suffice.
When $\alpha$ is a small constant, e.g.\ 0.1, what we find is that one can essentially reduce the number rows of the data structure from $\log(n)$ to just a few, e.g.\ one or two, at the cost of a factor $32^2=1024$ increase in space.\footnote{Recall, $\Omega(\log n/\log (1/\alpha))$ rows are necessary \countsketch whereas \bptree needs only $O(\log 1/\alpha)$ rows.}
This brief calculation suggests that \countsketch will outperform \bptree when the heaviness is $\alpha<1$ until $n\gtrsim 2^{1024}$---which is to say always in practice.
On the other hand, our experiments demonstrate that \hht clearly outperforms \countsketch with a sufficiently heavy heavy hitter.
More experimental work is necessary to determine the heaviness threshold (as a function of $n$) where \bptree out-performs \countsketch.
There are many parameters that affect the trade-offs among space, time, and accuracy, so such an investigation is beyond the scope of the preliminary results reported here.

\section{Conclusion}
In this paper we studied the heavy hitters problem, which is arguably one of the most important problems for data streams. The problem is heavily inspired from practice and algorithms for it are used in commercial systems. We presented the first space and time optimal algorithm for finding $\ell_2$-heavy hitters, which is the strongest notion of heavy hitters achievable in polylogarithmic space. By optimal, we mean the time is $O(1)$ to process each stream update, and the space is $O(\log n)$ bits of memory. These bounds match trivial lower bounds (for constant $\epsilon$). We also provided new techniques which may be of independent interest: (1) a one-pass implementation of a multi-round adaptive compressed-sensing scheme where we use that after filtering a fraction of items, the heavy item is becoming even heavier (2) a derandomization of Bernoulli processes relevant in this setting using limited independence. Both are essential in obtaining an optimal heavy hitters algorithm with $O(1)$ memory. Technique (1) illustrates a new power of insertion-only streams and technique (2) can be stated as a general chaining result with limited independence in terms of the size of the nets used. Given the potential practical value of such an algorithm, we provided preliminary experiments showing our savings over previous algorithms. 

\section*{Acknowledgments}
We would like to thank Huy L.\ Nguy$\tilde{\hat{\mbox{e}}}$n for pointing out that our original proof that $6$-wise independence suffices for the conclusion of Theorem~\ref{thm: sup bound} could be slightly tweaked to in fact show that $4$-wise independence suffices.
We would also like to thank the anonymous referees for their suggestions to improve the readability of this paper.

\bibliographystyle{abbrv}

\begin{thebibliography}{10}

\bibitem{AlonMS99}
N.~Alon, Y.~Matias, and M.~Szegedy.
\newblock The space complexity of approximating the frequency moments.
\newblock {\em J. Comput. Syst. Sci.}, 58(1):137--147, 1999.

\bibitem{BabcockBDMW02}
B.~Babcock, S.~Babu, M.~Datar, R.~Motwani, and J.~Widom.
\newblock Models and issues in data stream systems.
\newblock In {\em Proceedings of the Twenty-first {ACM} {SIGACT-SIGMOD-SIGART}
  Symposium on Principles of Database Systems (PODS)}, pages 1--16, 2002.

\bibitem{BerindeCIS09}
R.~Berinde, G.~Cormode, P.~Indyk, and M.~J. Strauss.
\newblock Space-optimal heavy hitters with strong error bounds.
\newblock In {\em Proceedings of the Twenty-Eigth {ACM} {SIGMOD-SIGACT-SIGART}
  Symposium on Principles of Database Systems (PODS)}, pages 157--166, 2009.

\bibitem{BhattacharyyaDW16}
A.~Bhattacharyya, P.~Dey, and D.~P. Woodruff.
\newblock An optimal algorithm for $\ell_1$-heavy hitters in insertion streams
  and related problems.
\newblock {\em CoRR}, abs/1511.00661, 2016.

\bibitem{BoyerM81}
R.~S. Boyer and J.~S. Moore.
\newblock A fast majority vote algorithm.
\newblock Technical Report Technical Report ICSCA-CMP-32, Institute for
  Computer Science, University of Texas, 1981.

\bibitem{BoyerM91}
R.~S. Boyer and J.~S. Moore.
\newblock {MJRTY:} {A} fast majority vote algorithm.
\newblock In {\em Automated Reasoning: Essays in Honor of Woody Bledsoe}, pages
  105--118, 1991.

\bibitem{BravermanCIW16}
V.~Braverman, S.~R. Chestnut, N.~Ivkin, and D.~P. Woodruff.
\newblock {Beating CountSketch for Heavy Hitters in Insertion Streams}.
\newblock In {\em Proceedings of the 48th Annual {ACM} Symposium on Theory of
  Computing (STOC), to appear}, 2016.
\newblock Full version at arXiv abs/1511.00661.

\bibitem{braverman2013approximating}
V.~Braverman and R.~Ostrovsky.
\newblock Approximating large frequency moments with pick-and-drop sampling.
\newblock In {\em Approximation, Randomization, and Combinatorial Optimization.
  Algorithms and Techniques}, pages 42--57. Springer, 2013.

\bibitem{CarterW79}
L.~Carter and M.~N. Wegman.
\newblock Universal classes of hash functions.
\newblock {\em J. Comput. Syst. Sci.}, 18(2):143--154, 1979.

\bibitem{ChakrabartiCM10}
A.~Chakrabarti, G.~Cormode, and A.~McGregor.
\newblock A near-optimal algorithm for estimating the entropy of a stream.
\newblock {\em {ACM} Transactions on Algorithms}, 6(3), 2010.

\bibitem{CharikarCF04}
M.~Charikar, K.~C. Chen, and M.~Farach{-}Colton.
\newblock Finding frequent items in data streams.
\newblock {\em Theor. Comput. Sci.}, 312(1):3--15, 2004.

\bibitem{CormodeM05}
G.~Cormode and S.~Muthukrishnan.
\newblock An improved data stream summary: the count-min sketch and its
  applications.
\newblock {\em J. Algorithms}, 55(1):58--75, 2005.

\bibitem{DemaineLM02}
E.~D. Demaine, A.~L{\'{o}}pez{-}Ortiz, and J.~I. Munro.
\newblock Frequency estimation of {Internet} packet streams with limited space.
\newblock In {\em Proceedings of the 10th Annual European Symposium on
  Algorithms (ESA)}, pages 348--360, 2002.

\bibitem{Dudley67}
R.~M. Dudley.
\newblock The sizes of compact subsets of {Hilbert} space and continuity of
  {Gaussian} processes.
\newblock {\em J. Functional Analysis}, 1:290--330, 1967.

\bibitem{GopalanR09}
P.~Gopalan and J.~Radhakrishnan.
\newblock Finding duplicates in a data stream.
\newblock In {\em Proceedings of the Twentieth Annual {ACM-SIAM} Symposium on
  Discrete Algorithms (SODA)}, pages 402--411, 2009.

\bibitem{Haagerup82}
U.~Haagerup.
\newblock The best constants in the {Khintchine} inequality.
\newblock {\em Studia Math.}, 70(3):231--283, 1982.

\bibitem{hanson1971bound}
D.~L. Hanson and F.~T. Wright.
\newblock A bound on tail probabilities for quadratic forms in independent
  random variables.
\newblock {\em The Annals of Mathematical Statistics}, 42(3):1079--1083, 1971.

\bibitem{HarveyNO08}
N.~J.~A. Harvey, J.~Nelson, and K.~Onak.
\newblock Sketching and streaming entropy via approximation theory.
\newblock In {\em 49th Annual {IEEE} Symposium on Foundations of Computer
  Science (FOCS)}, pages 489--498, 2008.

\bibitem{huang2014tracking}
Z.~Huang, W.~M. Tai, and K.~Yi.
\newblock Tracking the frequency moments at all times.
\newblock {\em arXiv preprint arXiv:1412.1763}, 2014.

\bibitem{IndykPW11}
P.~Indyk, E.~Price, and D.~P. Woodruff.
\newblock On the power of adaptivity in sparse recovery.
\newblock In {\em {IEEE} 52nd Annual Symposium on Foundations of Computer
  Science (FOCS)}, pages 285--294, 2011.

\bibitem{IndykW05}
P.~Indyk and D.~P. Woodruff.
\newblock Optimal approximations of the frequency moments of data streams.
\newblock In {\em Proceedings of the 37th Annual {ACM} Symposium on Theory of
  Computing (STOC)}, pages 202--208, 2005.

\bibitem{JowhariST11}
H.~Jowhari, M.~Sa\u{g}lam, and G.~Tardos.
\newblock Tight bounds for {Lp} samplers, finding duplicates in streams, and
  related problems.
\newblock In {\em Proceedings of the 30th ACM {SIGMOD-SIGACT-SIGART} Symposium
  on Principles of Database Systems (PODS)}, pages 49--58, 2011.

\bibitem{KarpSP03}
R.~M. Karp, S.~Shenker, and C.~H. Papadimitriou.
\newblock A simple algorithm for finding frequent elements in streams and bags.
\newblock {\em {ACM} Trans. Database Syst.}, 28:51--55, 2003.

\bibitem{li2013tight}
Y.~Li and D.~P. Woodruff.
\newblock A tight lower bound for high frequency moment estimation with small
  error.
\newblock In {\em Approximation, Randomization, and Combinatorial Optimization.
  Algorithms and Techniques}, pages 623--638. Springer, 2013.

\bibitem{MankuM12}
G.~S. Manku and R.~Motwani.
\newblock Approximate frequency counts over data streams.
\newblock {\em {PVLDB}}, 5(12):1699, 2012.

\bibitem{MetwallyAA05}
A.~Metwally, D.~Agrawal, and A.~{El Abbadi}.
\newblock Efficient computation of frequent and top-k elements in data streams.
\newblock In {\em Proceedings of the 10th International Conference on Database
  Theory (ICDT)}, pages 398--412, 2005.

\bibitem{MisraG82}
J.~Misra and D.~Gries.
\newblock Finding repeated elements.
\newblock {\em Sci. Comput. Program.}, 2(2):143--152, 1982.

\bibitem{MonemizadehW10}
M.~Monemizadeh and D.~P. Woodruff.
\newblock 1-pass relative-error {$L_p$}-sampling with applications.
\newblock In {\em Proceedings of the Twenty-First Annual {ACM-SIAM} Symposium
  on Discrete Algorithms (SODA)}, pages 1143--1160, 2010.

\bibitem{Muthukrishnan05}
S.~Muthukrishnan.
\newblock Data streams: Algorithms and applications.
\newblock {\em Foundations and Trends in Theoretical Computer Science}, 1(2),
  2005.

\bibitem{ThorupZ04}
M.~Thorup and Y.~Zhang.
\newblock Tabulation based 4-universal hashing with applications to second
  moment estimation.
\newblock In {\em Proceedings of the Fifteenth Annual {ACM-SIAM} Symposium on
  Discrete Algorithms, {SODA} 2004, New Orleans, Louisiana, USA, January 11-14,
  2004}, pages 615--624, 2004.

\bibitem{thorup2012tabulation}
M.~Thorup and Y.~Zhang.
\newblock Tabulation-based 5-independent hashing with applications to linear
  probing and second moment estimation.
\newblock {\em SIAM Journal on Computing}, 41(2):293--331, 2012.

\end{thebibliography}

\appendix

\section{\pickdrop counter example}\label{app: pick and drop}

We begin with a brief description of the \pickdrop algorithm and refer the reader to~\cite{braverman2013approximating} for the full details.
Afterwards we will describe a stream where the algorithm fails, with probability at least $1-n^{-1/8}$, to find a $\ell_2$ heavy hitter  and give some intuition about why the algorithm has this behavior. That is, the probability the algorithm succeeds is inverse polynomially small.

There is a parameter to the algorithm $t$ and the algorithm begins by partitioning the stream into $m/t$ consecutive intervals of $t$ updates each.
The value of $t$ is chosen such that $m/t$ is roughly the smallest frequency of a heavy hitter that we wish to find.
Hence, the average number of heavy hitter updates per interval is $\Omega(1)$.
In each interval, independently, a position $T\in\{1,2,\ldots,t\}$ is chosen at random and the item in the $T$th position within the interval is sampled. 
We also count how many times the sampled item it appears within $\{T,T+1,\ldots,t\}$. 
Next, the following ``competition'' is performed. 
We traverse the intervals sequentially from the first to the last and maintain a global sample and counter. 
Initially, the global sample is the sample from the first interval and the global counter is the counter from the first interval. 
For each future interval $i$, two options are possible: (1) the global sample is replaced with the sample from $i$ and the global counter is replaced with the counter from interval $i$, or (2) the global sample remains unchanged and the global counter is increased by the number of times the global sample item appears in interval $i$. 
Also, the algorithm maintains $X$ which is the current number of intervals for which the current global counter has not been replaced. 
If the maximum between $X$ and the counter from the $i$th interval is greater than the global counter then (1) is executed, otherwise (2) is executed.

Consider the following counter example that shows \pickdrop cannot find $\ell_2$ heavy hitters in $O(1)$ words.
$f\in\R^{2n}$ is a frequency vector where one coordinate $H$ has frequency $\sqrt{n}$, $n$ elements have frequency $1$, and $\sqrt{n}$ elements have frequency $n^{1/4}$, call the latter ``pseudo-heavy''. 
The remaining coordinates of $f$ are~0.
Consider the stream that is split into $t$ intervals $B_1,\ldots,B_t$ where $t =\sqrt{n}$ and each interval has size $\Theta(t)$. 
The items are distributed as follows.
\begin{itemize}\itemsep=-3pt
\item Each interval $w$ where $w = q n^{1/4}$ for $q = 1,2,\ldots, n^{1/4}$, is filled with $n^{1/4}$ pseudo-heavy elements each appearing $n^{1/4}$ times and appearing in no other interval.
\item Each interval $w+h$, for $h = 1,2,\ldots,n^{1/8}$ contains $n^{1/8}$ appearances of $H$ and remaining items that appear only once in the stream.
\item Each interval $w+h$, for $h=n^{1/8}+1,\ldots,n^{1/4}-1$ contains items that appear only once in the stream.
\end{itemize}
Obviously, a pseudo-heavy element will be picked in every ``$w$ interval''.
In order for it to be beaten by $H$, its count must be smaller than $n^{1/8}$ and $H$ must be picked from one of the $n^{1/8}$ intervals immediately following.
The intersection of these events happens with probability no more than $n^{-1/8}\left(n^{1/8}\cdot\frac{n^{1/8}}{n^{1/2}}\right) = n^{-3/8}$.
As there are only $n^{1/4}$ ``$w$ intervals'', the probability that the algorithm outputs $H$ is smaller than $n^{-1/8}$.

Note that the algorithm cannot be fixed by choosing other values of $t$ in the above example. If $t\gg \sqrt{n}$ then $H$ might be sampled with higher probability but the pseudo-heavy elements will also have higher probabilities and the above argument can be repeated with a different distribution in the intervals. If $t\ll \sqrt{n}$ then the probability to sample $H$ in any of the rounds becomes too small.

This counterexample is not contrived---it shows why the whole \pickdrop sampling algorithm fails to shed any light on the $\ell_2$ heavy hitters problem. 
Let us explain why the algorithm will not work in polylogarithmic space for $k = 2$. Consider the case when the global sample is $h\neq  H$ and the global counter is $f_h$. In this case, the global sample can ``kill'' $f_h$ appearances of $H$ in the next $f_h$ intervals, by the description of the algorithm. The probability to sample $h$ is $f_h/t$, so the expected number of appearances of $H$ that will be killed is upper bounded by $\sum_h f_h^3 / t = F_3/t$. In the algorithm, we typically choose $t = \sqrt{F_1}$. Consider the case when $F_1 = \Theta(n)$, $F_2 = \Theta(n)$  but $F_2 =o(F_3)$. In this case the algorithm needs $f_H$ to be at least $C F_3/\sqrt{F_2}$ for a constant $C$. This is impossible if $f^2_H = \Theta(F_2)$.
For $t = o(\sqrt{n})$ the probability that $H$ is sampled becomes $o(1)$. For $t = \omega(\sqrt{n})$ we need a smaller decay for $H$ to survive until the end in which case the above analysis can be repeated with the new decay factor for pseudo-heavy elements.

\section{Alternative $F_2$ tracking analysis}\label{app: another F2}

This section describes an alternative $F_2$ tracking analysis.
It shows that $4$-wise independence is enough, that is the original setting used by~\cite{AlonMS99}, but at the cost of an extra $\log 1/\epsilon$ factor on the space complexity.
Let $N\in\mathbb{N}$, $Z^j$ be a vector of $4$-wise independent Rademacher random variables for $j\in N$, and define $X_{j,t} = \langle Z^j,f^{(t)}\rangle$.
Let $Y_t = \frac{1}{N}\sum_{j=1}^NX_{j,t}^2$, obviously $Y_t$ can be computed by a streaming algorithm.

\begin{theorem}\label{thm: beefed up AMS 2}
Let $0<\epsilon<1$.
There is a streaming algorithm that outputs at each time $t$ a value $\hat{F}_2^{(t)}$ such that
\[\Pr(|\hat{F}_2^{(t)}-F_2^{(t)}|\leq \epsilon F_2,\text{ for all }0\leq t\leq m)\geq 1-\delta.\]
The algorithm use $O(\frac{1}{\epsilon^2}\log\frac{1}{\delta\epsilon})$ words of storage and has $O(\frac{1}{\epsilon^{2}}\log\frac{1}{\delta\epsilon})$ update time.
\end{theorem}

The proof of Theorem~\ref{thm: beefed up AMS} uses the following technical lemma that bounds the divergence the estimate over an entire interval of updates.
It follows along the lines of Lemma~22 from the full version of \cite{BravermanCIW16}.
\newcommand{\amsintervallemmatext}{Let $\epsilon<1$, $N\geq 12/\epsilon^2$, and $\Delta>0$.
If $F_2^{(v)}-F_2^{(u)}\leq (\frac{\epsilon}{20NC_*})^2F_2  $,  then
\[\Pr\left( Y_t-F^{(t)}_2\leq 2\epsilon F_2,\,\forall t\in [u,v]\right) \geq \frac{2}{3}.\]}
\begin{lemma}\label{lem: AMS interval}
\amsintervallemmatext
\end{lemma}
\begin{proof}
We denote $\Delta=(F_2^{(v)}-F_2^{(u)})/F_2$ and split the expression above as follows
\[Y_t-F^{(t)}_2 =  (Y_t-Y_u) + (Y_u-F^{(u)}_2) + (F^{(t)}_2-F^{(u)}_2).\]
Let $b_1>0$. 
For the second term, and it is shown in~\cite{AlonMS99} that
\[\Pr(Y_u-F^{(u)}_2\geq b_1)\leq \frac{2(F_2^{(u)})^2}{N b_1^2}.\]
Let $X^{(t)}=\frac{1}{\sqrt{N}}(X_{1,t},X_{2,t},\ldots,X_{N,t})$ and $X^{(u:t)}=X_t-X_u$, in particular $Y_t = \|X^{(t)}\|_2^2$.
For the first term, we have
\[Y_t = \|X^{(u)}+X^{(u:t)}\|^2_2 \leq \left(\|X^{(u)}\|_2 +\|X^{(u:t)}\|_2\right)^2,\]
so
\[Y_t-Y_u \leq 2\sqrt{Y_u}\|X^{(u:t)}\|_2 + \|X^{(u:t)}\|_2^2.\]
Now from Theorem~\ref{thm: sup bound} and a union bound it follows that for $b_2>0$
\[P\left(\sup_{j\in[N],u\leq t\leq v}|X_{j,t}-X_{j,u}|\geq b_2\right)\leq \frac{N C_*\|f^{(u:v)}\|_2}{b_2},\]
so $P(\sup_{t\geq u}\|X^{(u:t)}\|_2\geq b_2) \leq NC_*\|f^{(u:v)}\|_2/b_2$.

With probability at least $1-\frac{2(F_2^{(u)})^2}{N b^2_1}-\frac{NC_*\|f^{(u:v)}\|_2}{b_2}$ we get, for all $t\geq u$
\begin{align*}
Y_t-F_2^{(t)}
&\leq 2(F_2^{(u)}+b_1)^{\frac{1}{2}}b_2 + b_2^2 + b_1 + \Delta F_2.
\end{align*}
Now we set $b_1=\epsilon F_2$ and $b_2=6NC_*\sqrt{\Delta F_2}$ and the above expression is bounded by
\begin{align*}
Y_t - F_2^{(t)}&\leq 12\sqrt{2}NC_*F_2\sqrt{\Delta} + \epsilon F_2 +2\Delta\\
&\leq 20NC_*F_2\sqrt{\Delta} + \epsilon F_2,\\
&\leq 2\epsilon F_2
\end{align*}
since $\Delta\leq F_2$.
The probability of success is at least
\begin{align*}
1 - \frac{2(F_2^{(u)})^2}{N \epsilon^2F_2^2} - \frac{\|f^{(u:m)}\|_2}{6\sqrt{\Delta}} \leq 1-\frac{1}{3}.
\end{align*}
\end{proof}

We now proceed to describe the $F_2$ estimation algorithm and prove Theorem~\ref{thm: beefed up AMS}.
\begin{proof}[proof of Theorem~\ref{thm: beefed up AMS}]
  The algorithm returns at each time $t$ the value
\[\hat{F}_2(t) = \max_{s\leq t}\median(Y_{1,t},Y_{2,t},\ldots,Y_{M,t}),\] 
which are $M=\Theta(\log\frac{1}{\delta\epsilon})$ independent copies of $Y_t$ with $N=12(\frac{3}{\epsilon})^2$.
Let $\Delta = (\frac{\epsilon/3}{20NC_*})^2$ and define $t_0=0$ and $t_k = \max\{t\leq m\mid F_2^{(t_k)}\leq F_2^{(t_{k-1})}+\Delta F_2\}$, for $k\geq 1$.
These times separate the stream into no more than $2/\Delta$ intervals during which the second moment increases no more than $\Delta F_2$.
Lemma~\ref{lem: AMS interval} and an application of Chernoff's Inequality imply that for each interval $k$
\begin{align*}
\Pr&\left( \median_{j\in[M]}(Y_{j,t}) -F_2^{(t)}\leq \frac{2}{3}\epsilon F_2,\,\forall t\in[t_{k-1},t_k)\right)\\
&\geq 1-\poly(\delta\epsilon).
\end{align*}
The original description of the AMS algorithm~\cite{AlonMS99} implies that, for all $k$,
\[\Pr\left( \median_{j\in[M]}(Y_{j,t_k})\geq (1-\epsilon/3)F_2^{(t)}\right)\geq 1-\poly(\delta\epsilon).\]
By choosing the constants appropriately and applying a union bound over all $O(\epsilon^{-2})$ intervals and endpoints we achieve all of these events events occur with probability at least $1-\delta$.
One easily checks that this gives the desired guarantee.
\end{proof}

\section{Faster with less space}\label{app: speedup}

Maintaining all of the $O(\epsilon^{-2}\log\frac{1}{\epsilon\delta})$ instances of the $F_2$ tracking algorithm is a significant speed and memory bottleneck for \bptree.
For example, Section~\ref{sec: experiments} uses a tracker with thirty counters, which is a lot of additional storage. 
Furthermore, evaluating the four-wise hash functions is one of the main bottlenecks limiting the speed, so it is much slower to have update one $F_2$ tracker for every row of the datastructure. 
This section addresses that bottleneck by reducing the number of instances of $F_2$ tracking to only one. 
This is the most efficient way we have found to implement the algorithm.
It does not improve (nor degrade) the space complexity or update time for \bptree.
In this section we will describe the intuition behind the following theorem. 
Its proof appears in the full version of this paper.
\begin{theorem}\label{thm: faster bptree}
The algorithm \bptree can be implemented with a single $F_2$ tracker, which periodically restarts.
The storage, update time, and retrieval times all retain the same asymptotic complexity.
\end{theorem}

The single-tracker version will use the same data structure with $b\cdot r$ buctets, where $\frac{K^2}{\epsilon^2} < b = O(1/\epsilon^2)$ and $r=O(\log{1/\delta\epsilon})$.
However, the supressed constants may be larger.
For convenience, we will take $b\cdot r$ to be the dimensions of the auxiliary \countsketch, too.  
However, it is not necessary that the same size table is used for the auxiliary \countsketch.

We will make two modifications to \bptree.
The first is to change how we determine when to restart \hho.
Instead of using one $F_2$ tracker in each of the $b\cdot r = O(\epsilon^{-2}\log \frac{1}{\delta\epsilon})$ buckets, we will use only one $F_2$ tracker overall, which periodically restarts.
To be precise, replace all $b\cdot r=O(\epsilon^{-2}\log 1/\delta\epsilon)$ copies of the $F_2$ trackers used by the instances of \hht in each bucket with a single tracker that is periodically restarted.
At the beginning we start $T_1$, the first $F_2$ tracker, we label the first restart $T_2$, then $T_3$, and so on. 
The random bits used for each new tracker are independent of the previous ones.
For $k\geq 1$, we recursively define $t_k = \min\{t\geq 0\mid T_k^{(t)}\geq 2^k\}$ and $T_k^{(t)}$, $t\geq t_{k-1}$, is an $F_2$ estimate on the stream at times $t_{k-1},t_{k-1}+1,\ldots t$.
Let $\hat{\ell}$ be the index of the last tracker started, and we define $t_0=1$ and $t_{\hat{\ell}}=m$ for convenience even though $T_{\hat{\ell}}^{(t)}<2^{\hat{\ell}}$ by definition.
We will take the trackers to each have error probability $\delta'=O(\delta/\log \epsilon^{-1})$, where the supressed constant will be taken sufficiently small, and accuracy $\epsilon'=1/100$.
Upon each time $t_k$, for $k=0,1,2,\ldots,\hat{\ell}$, the instance of \hho in every bucket is restarted simultaneously with the value $\sigma=\frac{\epsilon}{16}2^{k/2}$ given as the input ``guess'' of $F_2$.
When a restart happens at time $t_k$, the old $F_2$ tracker $T_k$ and the corresponding \hho instance are discarded by the algorithm.
Here is another difference from the implementation using \hht, this faster version does not stagger-start the instances of \hho, only a single copy of \hho is operating in each bucket at one time.

The second modification is to guarantee that we do not discard a heavy hitter after it is identified.
Let us focus on a singe bucket.
In this bucket, \hho is restarted repeatedly, and each time it could output a candidate heavy hitter.
Rather than keep the last two candidates, as in \bptree, we maintain the identity of one candidate, call it $H_0$, and each time a new candidate is identified for this bucket, we compare its estimated frequency with an estimate of $f_{H_0}$ and replace $H_0$ with the new candidate if only if the freqency estimate of $H_0$ is the smaller of the two.
When comparing frequencies we use estimates from the auxiliary \countsketch.

The rest of this section contains the details of the proof of Theorem~\ref{thm: faster bptree}, but here is a rough outline.
We define two favorable events, these are $E_1$:``the $F_2$ tracking is accurate'' and $E_2$:``the auxiliary \countsketch is accurate'' and show that they have high enough probability in Lemmas~\ref{lem: E1} and~\ref{lem: E2}.
Next, we show that the probability that some heavy hitter $H$ is identified, conditionally given the random bits used by the $F_2$ trackers and $E\coloneqq E_1\cap E_2$, is $\Omega(1)$ per each of the $r$ rows of the data structure (Lemmas~\ref{lem: i and ii} and~\ref{lem: iii}).
That is enough to guarantee that $H$ is identified at least once among all $r=O(\log 1/\delta\epsilon)$ rows with probability at least $1-\poly(\delta\epsilon)$ using a Chernoff Bound, so we can apply a union bound over all heavy hitters. 
A more detailed version of that argument completes the proof of Theorem~\ref{thm: faster bptree} at the end of this section.

We need the conditioning because the rows of the \bptree data structure are no longer independent, since all of the restarts are timed from the same $F_2$ tracker, and that precludes a Chernoff Bound.
However, the rows are conditionally independent.
In order to show that $H$ is indeed identified with probability $\Omega(1)$, we will define the event that a restart $k$ is ``good'' in a particular bucket of the data structure.  
If $k$ is good it means two things.
First, the corresponding interval of updates, $t_{k},t_{k}+1,\ldots,t_{k+1}-1$, in this bucket satisfies the hypotheses of Lemma~\ref{lem: R rounds}, so $H$ is identified with (conditional) probability at least 2/3, and second, $H$ is never eclipsed as the heavy hitter in any of the subsequent rounds.
The rationale behind the proof of the bound itself is that, on average, rounds that are not good use up few of the updates to $H$, because the correspond to intervals in the stream where $H$ is not a heavy hitter.

Let $\calT_k$ denote the random bits used for all trackers started at times $t_{k'}$, for $k'\leq k$, and let $\calT$ denote the random bits for all of the trackers in total.
Similarly, let $\calC$ denote the collection of random bits used by the auxiliary \countsketch.
In most of the events we need to consider we can as well condition on $\calC$ without any extra legwork because $\calC$ is independent of $\calT$ and all instances of \hho.
Let $t^* = \min\{t\geq0\mid F_2^{(t)}\geq \frac{\epsilon^2}{100^2}F_2\}$ and let $\ell^* = \min\{k\mid t_k\geq t^*\}$.
We define $E_1$ to be the event that, for all $k$ where $t_k\geq t^*$ and all $t\in[t_{k-1},t_k)$, $|\hat{T}_k^{(t)} - F_2^{(t_{k-1}:t)}|\leq \frac{1}{100}F_2^{(t_{k-1}:t_k)}$.
\begin{lemma}\label{lem: E1}
$\Pr(E_1)\geq 1-\delta/3$
\end{lemma}
\begin{proof}
Let $D_k$ be the event that, for all $t\in[t_{k-1},t_k)$, $|\hat{T}_k^{(t)} - F_2^{(t_{k-1}:t)}|\leq \frac{1}{100}F_2^{(t_{k-1}:t_k)}$.
By Theorem~\ref{thm: beefed up AMS} and our choice $\delta'=O(\delta/\log\epsilon^{-1})$ for the error probability of the $F_2$ tracking, we have that, for all $k=1,2,\ldots,\hat{\ell}+1$, $\Pr(D_k\mid\calT_{k-1})\geq 1-\delta'$.
Starting from $t^*$, the value of $F_2^{(t)}$ doubles only $O(\log \epsilon^{-1})$ times in the remainder of the stream.
Therefore, $\hat{\ell}-\ell^*=O(\log\epsilon^{-1})$ on the event $E_1$, and $E_1\subseteq \bigcup_{k=\hat{\ell}-O(\log\frac{1}{\epsilon})}^{\hat{\ell}+1} D_k$.
The bound follows as
\begin{align*}
1-\Pr(E_1) 
&\leq \sum_{k\in L} 1-\Pr(D_k)\\
&= \sum_{k\in L} 1-\E\Pr(D_k\mid\calT_{k-1})\\
&\leq O(\log\epsilon^{-1})\cdot \delta'\leq\frac{\delta}{3}.
\end{align*}
The last inequality follows by taking the supressed constant in the definition of $\delta'$ to be sufficiently small.
\end{proof}
An observation from the last proof that we will use again is $E_1$ implies $O(\log \epsilon^{-1})$ restarts are enough to go from $t^*$ to the start of the stream.
Let $L=\{\ell^*,\ell^*+1,\ldots,\hat{\ell}\}$.
Although $L$ is random ($\hat{\ell}$ and $\ell^*$ depend on the $F_2$ trackers), $|L|=O(\log \epsilon^{-1})$ over the event $E_1$, which is sufficient for our purposes.

Now we will define the event $E_2$ that the auxiliary \countsketch frequency estimates are accurate enough.
Let $\calH$ denote the (random) set of candidate heavy hitters returned from a given bucket by instances of \hho started at times $t_k$, for $k\in L$, i.e. the final $O(\log\epsilon^{-1})$ instances of \hho run on the items in this bucket.
Let $\calH_{tot}$ denote the union over all of the buckets of the sets $\calH$.
$E_2$ is the event that at every time $t_{k}$, $k\in L$, and every $H'\in\calH_{tot}$, the auxiliary \countsketch returns a $(1\pm 1/2)$ approximation to $f_{H'}^{(t_{k})}$ when it is queried at time $t_{k}$.  
\begin{lemma}\label{lem: E2}
$\Pr(E_2|E_1)\geq 1-\delta/3$.
\end{lemma}
\begin{proof}
The number of buckets is $b r=O(\epsilon^{-2}\log 1/\delta\epsilon)$, in each bucket the number of candidate heavy hitters returned by the $L$ instances of \hho is $|\calH|\leq |L| = O(\log \epsilon^{-1})$, and we are requesting that the \countsketch gives an accurate estimate of each of those items on $|L|=O(\log \epsilon^{-1})$ different queries.
Thus we need a union bound over at most $|\calH_{tot}|\cdot|L| \leq b\cdot r\cdot|L|^2 = \poly(1/\epsilon)$ queries to the \countsketch.
The set $\calH_{tot}$ and the times $t_k$, $k\in L$, are independent of the auxiliary \countsketch, and since the auxiliary \countsketch has $r=O(\log 1/\delta\epsilon)$ rows we can guarantee that $P(E_2|E_1)\geq 1-\delta/3$ by choosing the surpressed constant on $r$ to be sufficiently large.
\end{proof}

Let us narrow our focus to a single bucket within the data structure and suppose it contains some $\epsilon$-heavy hitter $H$.  
Let $B\subseteq [n]\setminus\{H\}$ be the (random) set of items in the bucket besides $H$.
Let us say that index $k\in L$ is \emph{good} in this bucket if the following three things happen, where $K=2^{12}C_*$, 
\begin{enumerate}[(i)]\setlength{\itemsep}{-3pt}
\item $f_H^{(t_{k-1}:t_{k})}\geq \frac{\epsilon}{16} 2^{k/2}$,\label{it: K-heavy}
\item for all indices $k'\geq k$ and items $i\in B$, we have $f_H^{(t_{k'})}\geq 4 f_i^{(t_{k'})}$, and \label{it: good freq}
\item $F_2^{(t_{k-1}:t_{k})}(B)^{1/2}\leq \frac{\epsilon}{16K} 2^{k/2}$.\label{it: good sigma}
\end{enumerate}

The next two lemmas prove that within any bucket holding a heavy hitter there is a good index with (conditional) probability at least 1/2.
The first lemma establishes properties \eqref{it: K-heavy} and \eqref{it: good freq}.
The second lemma establishes~\eqref{it: good sigma}.

\begin{lemma}\label{lem: i and ii}
If $H$ is an item with frequency $f_H\geq \epsilon\sqrt{F_2}$ and is $B$ the set of items going into some bucket with $H$, then conditionally given $\calT$ and $E$ the probability that there exists an index $k^*\in L$ satisfying both \eqref{it: K-heavy} and \eqref{it: good freq} it at least 3/4.
\end{lemma}
\begin{proof}
Let $j\in L$ be the largest index such that exists an element $i\in B$ such that $4f_i^{(t_{j})}>f_H^{(t_{j})}$, or $j=\ell^*-1$ if no such index exists.  
Then $f_H^{(t_j)}\leq 4f_i^{(t_j)} \leq 4\sqrt{F_2(B)}$.
Let $F$ be the event that there is no $k^*\in L$ satisfying both \eqref{it: K-heavy} and \eqref{it: good freq}.
$F$ implies $f_H^{(t_{k-1}:t_{k})}<\frac{\epsilon}{16}2^{k/2}$, for all $k>j$.
Therefore, on the event $F$ we have
\begin{align*}
f_H &= f_H^{(t_j)} + \sum_{k=j+1}^{\hat{\ell}} f_H^{(t_{k-1}:t_k)} 
\leq 4\sqrt{F_2(B)} + \sum_{k=j+1}^{\hat{\ell}} \frac{\epsilon}{16} 2^{k/2}\\
&\leq 4\sqrt{F_2(B)} + \frac{\epsilon}{16} 2^{\hat{\ell}/2}\sum_{k=0}^{\infty} 2^{k/2}
= 4\sqrt{F_2(B)} + \frac{\epsilon}{4} 2^{\hat{\ell}}.
\end{align*}
The event $E$ implies $4F_2\geq 2^{\hat{\ell}}$, so $f_H-\frac{\epsilon}{4}2^{\hat{\ell}/2} 
\geq \frac{\epsilon}{2}\sqrt{F_2}$.
Thus,
$F_2(B)\geq \frac{\epsilon^2}{32} F_2$, and it implies $\Pr(F|\calT,E)\leq \Pr(F_2(B)\geq \frac{\epsilon^2}{32}F_2|\calT,E)$.
However, $\E(F_2(B)|\calT,E)=\E(F_2(B)) = \frac{1}{b}(F_2-f_H^2)$. 
The conditional probability can be bounded with Markov's Inequality to find $\Pr(F_2(B)\geq \frac{\epsilon^2}{32}F_2|\calT,E)\leq \frac{16}{\epsilon^2b}$, so $b = O(1/\epsilon^2)$ can be chosen large enough to guarantee $\Pr(F|\calT,E)\leq 1/4$, as desired.
\end{proof}

\begin{lemma}\label{lem: iii}
For every $k\in L$
\[\Pr\left(F_2^{(t_{k-1}:t_{k})}(B)
\leq \frac{\epsilon^2}{16K}2^{k}
\mid \calT,E\right)\geq 3/4.\]
\end{lemma}
\begin{proof}
We have
\begin{align*}
 \E(F_2^{(t_{k-1}:t_{k})}(B)|\calT,E)
&=\frac{1}{b}\E(F_2^{(t_{k-1}:t_{k})}-(f_H^{(t_{k-1}:t_{k})})^2|\calT,E)\\
&\leq\frac{1}{b}\E(F_2^{(t_{k-1}:t_{k})}|\calT,E)\\
&\leq \frac{1}{b}2^{k}(1+1/100).
\end{align*}
Upon choosing $b=O(1/\epsilon^2)$ sufficiently large, Markov's Inequality shows that the desired inequality holds with conditional probability at least $3/4$.
\end{proof}

\begin{proof}[Proof of Theorem~\ref{thm: faster bptree}]
It is enough to show that if $H\in[n]$ has freqency $f_H\geq\epsilon\sqrt{F_2}$ then in each repetition $r'\leq r$ the probability, conditionally given $\calT$, $\calC$, and $E\coloneqq E_1\cap E_2$, that $H$ is identified by some instance of \hho and is not discarded thereafter is $\Omega(1)$.  
The reason that this enough is that the $r=\Theta(\log1/\delta\epsilon)$ rows of the data structure are independent conditionally given $\calT$, $\calC$, and $E$, so a Chernoff's Bound guarantees that $H$ is found with probability $1-\delta\epsilon^2/3$.
Let $A$ represent the event that all $\epsilon$-heavy hitters are returned. 
By a union bound, what we have is $\Pr(A|\calT,\calC,E)\geq 1-\delta/3$.
Thus, by Lemmas~\ref{lem: E1} and~\ref{lem: E2}, $\Pr(A) \geq \Pr(A|E)-\Pr(\bar{E}) \geq \E (\Pr(A|\calT,\calC,E)) - 2\delta/3\geq 1-\delta$.

It remains to prove that within a given bucket containing a heavy hitter $H$ the probability that $H$ is reported from the bucket is $\Omega(1)$.
Conditioning on $\calC$, as well, doesn't change Lemmas~\ref{lem: i and ii} and~\ref{lem: iii}. 
Thus, the probability that an index is good conditionally given $\calT$, $\calC$, and $E$ is at least $1-2\cdot 1/4 =1/2$.
If index $k$ is good then $f_H\geq \sigma=\frac{\epsilon}{16}2^{k/2}\geq K(F_2^{(t_{k-1}:t_k)}(B))^{1/2}$.
Therefore, the hypotheses for Lemma~\ref{lem: R rounds} are met and $H$ is returned by the modified algorithm with (conditional) probability at least $2/3$.
Furthermore, after $H$ is identified it is never discarded because the estimated frequencies from the auxiliary \countsketch are guaranteed by \eqref{it: good freq} and $E$ to satisfy $\hat{f}^{(t_{k'})}_H> \hat{f}^{(t_{k'})}_j$, for all $j\in B$ and $k'\geq k$.
The probability that $H$ is identified, conditionally given $\calC$, $\calT$, and $E$ is at least $2/3-1/2=1/6$, which completes the proof.
\end{proof}

\end{document}